\documentclass[12pt]{article}
\usepackage{amsmath}
\usepackage{amssymb}
\usepackage{graphicx}
\usepackage{enumerate}
\usepackage{natbib}
\usepackage{amsthm}
\usepackage{url} 
\usepackage{mathtools}
\usepackage{mathrsfs}
\usepackage{eufrak}
\usepackage{bbm}
\usepackage{threeparttable}
\usepackage{footnote}
\usepackage{tablefootnote}
\usepackage{adjustbox}
\usepackage{makecell}
\usepackage{multirow}
\usepackage{dsfont}
\usepackage{amsfonts}
\usepackage[ruled,linesnumbered]{algorithm2e}
\usepackage{bm}
\usepackage{subfig}
\usepackage[normalem]{ulem}
\usepackage[title]{appendix}
\usepackage{textcomp}
\DeclareMathAlphabet{\mathpzc}{OT1}{pzc}{m}{it}
\usepackage{tikz}
\usetikzlibrary{arrows,chains,matrix,positioning,scopes}
\usepackage{hyperref}
\hypersetup{colorlinks=true,linkcolor=blue, citecolor=blue, linktocpage}
\mathtoolsset{showonlyrefs}
\allowdisplaybreaks
\newcommand{\blind}{1}

\makeatletter
\tikzset{join/.code=\tikzset{after node path={%
\ifx\tikzchainprevious\pgfutil@empty\else(\tikzchainprevious)%
edge[every join]#1(\tikzchaincurrent)\fi}}}
\makeatother

\tikzset{>=stealth',every on chain/.append style={join},
         every join/.style={->}}
\tikzstyle{labeled}=[execute at begin node=$\scriptstyle,
   execute at end node=$]

\makeatletter
\newcommand*{\Rom}[1]{\expandafter\@slowromancap\romannumeral #1@}
\makeatother
\newcommand*{\rom}[1]{\romannumeral #1}

\newcommand{\p}{\textup{P}}
\newcommand{\bp}{\textbf{\textup{P}}}
\newcommand{\tp}{\mathbb{P}}
\newcommand{\e}{\textup{E}}
\newcommand{\be}{\textbf{\textup{E}}}
\newcommand{\te}{\mathbb{E}}

\newcommand{\ty}{\mathcal{T}}

\newcommand{\infnorm}[1]{\|#1\|_{\infty}}
\newcommand{\twonorm}[1]{\|#1\|_{2}}

\newcommand{\maxnorm}[1]{\|#1\|_{\max}}

\newcommand{\norm}[1]{|#1|}

\newcommand{\bx}{\bm{x}}

\newcommand{\bmu}{\bm{\mu}}

\newcommand{\heta}{\hat{\eta}}

\newcommand{\supp}{\sup_{S: |S| \leq D}}

\newcommand{\cri}{\textup{Cr}}

\newcommand{\sj}[1]{S^{(j)}_{#1}}
\newcommand{\smj}[1]{S^{(-j)}_{#1}}

\def\boxit#1{\vbox{\hrule\hbox{\vrule\kern6pt\vbox{\kern6pt#1\kern6pt}\kern6pt\vrule}\hrule}}

\newcommand{\mcs}{[1]}
\newcommand{\mcf}{[0]}
\newcommand{\ind}{\perp\!\!\!\!\perp}

\newtheorem{example}{Example} 
\newtheorem{theorem}{Theorem}
\newtheorem{lemma}{Lemma} 
\newtheorem{proposition}{Proposition} 
\newtheorem{remark}{Remark}

\newtheorem{definition}{Definition}

\newtheorem{assumption}{Assumption}

\begin{document}

\def\spacingset#1{\renewcommand{\baselinestretch}%
{#1}\small\normalsize} \spacingset{1}


\if1\blind
{
  \title{\bf RaSE: A Variable Screening Framework via Random Subspace Ensembles}
  \author{Ye Tian\\
    Department of Statistics\\ Columbia University\\
    and \\
    Yang Feng \\
    Department of Biostatistics, School of Global Public Health\\New York University}
    \date{}
  \maketitle
} \fi

\if0\blind
{
  \bigskip
  \bigskip
  \bigskip
  \begin{center}
    {\LARGE\bf RaSE: A Variable Screening Framework via Random Subspace Ensembles}
\end{center}
  \medskip
} \fi

\bigskip
\begin{abstract}
Variable screening methods have been shown to be effective in dimension reduction under the ultra-high dimensional setting. Most existing screening methods are designed to rank the predictors according to  their individual contributions to the response. As a result,  variables that are marginally independent but jointly dependent with the response could be missed. In this work, we propose a new framework for variable screening, \emph{Random Subspace Ensemble (RaSE)},  which works by evaluating the quality of random subspaces that may cover multiple predictors. This new screening framework can be naturally combined with any subspace evaluation criterion, which leads to an array of screening methods. The framework is capable to identify signals with no marginal effect or with high-order interaction effects. It is shown to enjoy the sure screening property and rank consistency. We also develop an iterative version of RaSE screening with theoretical support. Extensive simulation studies and real-data analysis show the effectiveness of the new screening framework.
\end{abstract}

\noindent%
{\it Keywords:} Variable screening; Random subspace method; Ensemble learning; Sure screening property; Rank consistency; High dimensional data; Variable selection

\spacingset{1.5} 

\section{Introduction}\label{sec:intro}
With the rapid advancement of computing power and technology, high-dimensional data become ubiquitous in many disciplines such as genomics, image analysis, and tomography. With high-dimensional data, the number of variables $p$ could be much larger than the sample size $n$. What makes statistical inference possible is the sparsity assumption, which assumes only a few variables have contributions to the response. Under this sparsity assumption, there has been a rich literature on the topic of variable selection, including LASSO \citep{tibshirani1996regression}, SCAD \citep{fan2001variable}, elastic net \citep{zou2005regularization}, and MCP \citep{zhang2010nearly}. Despite the success of these methods in many applications, for the ultra-high dimensional scenario where the dimension $p$ grows exponentially with $n$,  they may not work well due to the ``curse of dimensionality" in terms of simultaneous challenges to computational expediency, statistical accuracy, and algorithmic stability \citep{fan2009ultrahigh}.

To conquer these difficulties, \cite{fan2008sure} proposed a novel procedure called \emph{sure independence screening}  (SIS) with solid theoretical support. In the past decade, the power of feature screening has been well recognized and a myriad of screening methods have been proposed. The existing screening methods can be broadly classified into two categories, model-based methods and model-free ones. Model-based screening methods rely on specific models, such as SIS \citep{fan2008sure} and its extensions to generalized linear models \citep{fan2009ultrahigh}, Cox model \citep{fan2010high, zhao2012principled}, non-parametric independence screening method based on additive models \citep{fan2011nonparametric, cheng2014nonparametric} and screening via high-dimensional ordinary least-square projection (HOLP) \citep{wang2016high}. Recently, model-free approaches become more popular because of less stringent requirements. Examples of such approaches include the sure independent ranking and screening (SIRS) \citep{zhu2011model}, the screening method based on distance correlation (DC-SIS) and its iterative version \citep{li2012feature, zhong2015iterative}, screening procedure via martingale difference correlation (MDC-SIS) \citep{shao2014martingale}, screening via Kolmogorov filter \citep{mai2013kolmogorov, mai2015fused}, the screening approach for discriminant analysis (MV-SIS) \citep{cui2015model}, interaction pursuit via Pearson correlation (IP) and the distance correlation (IPDC) \citep{fan2016interaction, kong2017interaction}, the screening method based on ball correlation \citep{pan2018generic}, the nonparametric screening under conditional strictly convex loss \citep{han2019nonparametric}, and the screening method via covariate information number (CIS) \citep{nandy2020covariate}.

For variables that are marginally independent but jointly dependent with the response, many existing screening methods could miss them. This issue has been recognized in the literature \citep{fan2008sure,fan2009ultrahigh,zhu2011model,zhong2015iterative} and iterative screening procedures were developed, which were shown to be effective empirically. However, to the best of our knowledge, there is not much theoretical development for the iterative screening methods. In addition, some iterative screening methods (e.g.  iterative SIS) are coupled with a variable selection method like LASSO or SCAD, making its performance dependent on the specific choice of the regularization method. Besides, some other iterative screening methods (e.g. iterative SIRS and iterative DC-SIS) recruit variables step by step through residuals until a pre-specified number of variables are picked. Thus, their success hinges on a key tuning parameter, that is, how many variables to recruit in each step, making these procedures potentially less robust. 

These issues mentioned above motivate us to propose a new screening framework which goes  beyond marginal utilities. In the new framework, we investigate multiple features at the same time, via the random subspace method \citep{ho1998random}. \cite{tian2021rase} proposed a new \emph{Random Subspace Ensemble} classification method based on a similar idea, \textit{RaSE}, according to a specific aggregation framework first introduced in \cite{cannings2017random}. They advocated applying RaSE on sparse classification problems. The main idea of RaSE can be simply described as follows. First, $B_1B_2$ random subspaces are generated from a specific distribution on subspaces, which are evenly divided into $B_1$ groups. Next, the best subspace within each group is picked according to some criterion and a base learner is trained in that subspace. Hence we obtain $B_1$ base learners, each of which corresponds to a subspace. Finally, these $B_1$ base learners are aggregated on average and the ensemble will be used in prediction. The vanilla RaSE algorithm is reviewed in Algorithm \ref{algo: rase} in Appendix \ref{subsec: rase algo}. It's important to note that there is a by-product of RaSE, which is the selected proportion of each variable within $B_1$ selected subspaces. In this work, we will use this selected proportion to do variable screening, and call this the \emph{RaSE screening} framework.

We highlight the merits of RaSE screening framework as follows. First, by looking at different feature subspaces, variables marginally independent but jointly dependent with the response can be identified. Second, instead of proposing only a single screening approach, the flexible framework of RaSE allows us to use any criterion function for comparing subspaces, leading to an array of screening methods. One possible way to construct such a criterion function is to choose a base learner and a specific measure for comparing the subspaces. For example, if we know linear methods are suitable for the data, then we can apply RaSE by picking subspaces achieving lower BIC under linear models. If $k$-nearest neighbor ($k$NN) is believed to perform better, we can apply RaSE by choosing subspaces with the smallest cross-validation error on $k$NN. Third, under general conditions, we show the sure screening property and rank consistency for RaSE screening framework.  
Finally, we develop a novel iterative RaSE screening framework with sure screening property established  without the need to use a variable selection step or specifying the number of variables to recruit in each step.

The rest of this paper is organized as follows. Section \ref{sec:rase} introduces the vanilla RaSE screening framework and its iterative version in detail, and discusses  the relationship between RaSE and marginal screening methods. In Section \ref{sec:theory}, we present the theoretical properties for vanilla RaSE and iterative RaSE screening, including sure screening property and rank consistency. In Section \ref{sec:numerical}, extensive simulation studies and real-data analysis are conducted to demonstrate the power of our new screening framework. We summarize our contributions and point out some promising future avenues in Section \ref{sec:discussion}. The supplementary materials include all the technical proofs as well as additional details.

\section{RaSE: A General Variable Screening Framework}\label{sec:rase}
In what follows, we consider predictors $\bx = (x_1, \ldots, x_p)^T$ and response $y$. For regression problems, $y$ takes value from the real line $\mathbb{R}$, while for classification problems, $y$ takes value from an integer set $\{1, \ldots, K\}$, where $K > 1$ is a known integer. Denote the training data as $\{(\bx_i, y_i)\}_{i=1}^n$. Denote by $S_{\textup{Full}} = \{1, \ldots, p\}$ the full feature set. The signal set $S^* \subseteq S_{\textup{Full}}$ is defined as the set $S$ with minimal cardinality satisfying $y|\bx_{S} \ind \bx_{S_{\textup{Full}}\backslash S}$.  Denote $p^* = \norm{S^*}$. $[a]$ is used to represent the largest integer no larger than $a$.

To introduce the RaSE framework, we denote the $B_1B_2$ random subspaces as $\{S_{b_1 b_2}, b_1 = 1\ldots, B_1, b_2=1\ldots, B_2\}$, the $b_1$-th group of subspaces as $\{S_{b_1 b_2}\}_{b_2=1}^{B_2}$, and the selected $B_1$ subspaces as $\{S_{b_1*}\}_{b_1=1}^{B_1}$. The objective function corresponding to the specific criterion to choose subspaces is written as $\cri_n: \mathcal{S} \rightarrow \mathbb{R}$, where $\mathcal{S}$ is the collection of all subspaces. Assume a smaller value of $\cri_n$ leads to a better subspace.  Although the original RaSE \citep{tian2021rase} was introduced to solve classification problems, we now consider the general prediction framework, including both classification and regression.



\subsection{Vanilla RaSE screening framework}
Following the idea of \cite{tian2021rase}, we use the proportion of each feature among the selected $B_1$ subspaces as the importance measure.  Therefore, a natural screening procedure is to rank variables based on this proportion vector, then pick the variables with the largest proportions. The RaSE screening framework is summarized in Algorithm \ref{algo: rase screening}.
 
\begin{algorithm}
\caption{Vanilla RaSE screening}
\label{algo: rase screening}
\KwIn{training data $\{(\bm{x}_i, y_i)\}_{i = 1}^n$, subspace distribution $\mathcal{D}$, criterion function $\cri_n$, integers $B_1$ and $B_2$, number of variables $N$ to select}
\KwOut{the selected proportion of each feature $\hat{\bm{\eta}}$, the selected subset $\hat{S}$}
Independently generate random subspaces $S_{b_1b_2} \sim \mathcal{D}, 1 \leq b_1 \leq B_1, 1 \leq b_2 \leq B_2$\\
\For{$b_1 \leftarrow 1$ \KwTo $B_1$}{
  Select the optimal subspace $S_{b_1*} = S_{b_1b_2^*}$, where $b_2^*= \arg\min\limits_{1 \leq b_2 \leq B_2} \cri_n(S_{b_1b_2})$ \\
}
Output the selected proportion of each feature $\hat{\bm{\eta}}=(\heta_1,\ldots,\heta_p)^T$ where $\heta_j=B_1^{-1}\sum_{b_1=1}^{B_1}\mathds{1}(j\in S_{b_1*}), j=1,\ldots,p$\\
Output $\hat{S} = \{1 \leq j \leq p: \heta_j \textup{ is among the }  N \textup{ largest of all}\}$
\end{algorithm}

In the algorithm, the subspace distribution $\mathcal{D}$ is chosen as a \emph{hierarchical uniform distribution} over the subspaces by default. Specifically, with $D$ as the upper bound of the subspace size, we first generate the subspace size $d$ from the uniform distribution over $\{1, \ldots, D\}$. Then, the subspace $S_{11}$ follows the uniform distribution over all size-$d$ subspaces $\{S \subseteq S_{\textrm{Full}}: |S| = d\}$.  In practice, the subspace distribution can be adjusted if we have prior information about the data structure.

Algorithm \ref{algo: rase screening} is not the end of the story because it ranks all the variables but does not determine how many variables to keep. To facilitate the theoretical analysis, we define the final feature subset to be selected as
\begin{equation}\label{S alpha hat}
	\hat{S}_{\alpha} = \{1 \leq j \leq p: \heta_j \textup{ is among the } \allowbreak [\alpha D/c_{2n}] \textup{ largest of all}\},
\end{equation}
where $c_{2n}$ is a constant (to be specified in the next section) depending on $n$, $B_2$, $D$, and the criterion $\cri$ which is a population counterpart of $\cri_n$. Here, $\alpha$ can be any constant larger than 1, which will appear in the upper bound introduced in the sure screening theorem of Section \ref{sec:theory}. 

\subsection{Iterative RaSE screening}

As we mentioned in the introduction, the existing iterative screening methods have various tuning components such as the number of variables to recruit in each step and/or a specific variable selection method. We propose the iterative RaSE screening in Algorithm \ref{algo: iterative rase screening} to tackle these issues. 

The main idea of iterative RaSE screening is to update the subspace distribution based on the selected proportion in the preceding steps and not to conduct variable screening until the final step. To understand the details in the algorithm, we introduce a new subspace distribution.

\begin{algorithm}[!h]
\SetAlgoLined
\caption{Iterative RaSE screening ($\mbox{RaSE}_T$)}
\label{algo: iterative rase screening}
\KwIn{training data $\{(\bm{x}_i, y_i)\}_{i = 1}^n$, initial subspace distribution $\mathcal{D}^{[0]}$, criterion function $\cri_n$, integers $B_1$ and $B_2$, the number of iterations $T$, positive constant $C_0$, number of variables $N$ to select}
\KwOut{the selected proportion of each feature $\hat{\bm{\eta}}^{[T]}$, the selected subset $\hat{S}$}
\For{$t\leftarrow 0$ \KwTo $T$}{
	Independently generate random subspaces $S_{b_1b_2}^{[t]} \sim \mathcal{D}^{[t]}, 1 \leq b_1 \leq B_1$, $1 \leq b_2 \leq B_2$\\
	\For{$b_1 \leftarrow 1$ \KwTo $B_1$}{
  		Select the optimal subspace $S_{b_1*}^{[t]} =S_{b_1b_2^*}^{[t]}$, where $b_2^*= \arg\min\limits_{1 \leq b_2 \leq B_2} \cri_n(S_{b_1b_2}^{[t]})$ \\

	}
	Update $\hat{\bm{\eta}}^{[t]}$ where $\heta_j^{[t]}=B_1^{-1}\sum_{b_1=1}^{B_1}\mathds{1}(j\in S_{b_1*}^{[t]}), j=1,\ldots,p$\\
Update $\mathcal{D}^{[t+1]} \leftarrow$ hierarchical restrictive multinomial distribution $\mathcal{R}(\mathcal{U}_0, p, \tilde{\bm{\eta}}^{[t]})$, where $\tilde{\eta}_j^{[t]} \propto [\heta_j^{[t]}\mathds{1}(\heta_j^{[t]} > C_0/\log p) + \frac{C_0}{p}\mathds{1}(\heta_j^{[t]} \leq C_0/\log p)]$ and $\sum_{j=1}^p  \tilde{\eta}_j^{[t]} = 1$
}
Output the selected proportion of each feature $\hat{\bm{\eta}}^{[T]}$\\
Output $\hat{S} = \{1 \leq j \leq p: \heta_j^{[T]} \textup{ is among the }  N \textup{ largest of all}\}$
\end{algorithm}
Note that each subspace $S$ can be equivalently represented as $\bm{J}=(J_1,\ldots, J_p)^T$,  where $J_j = \mathds{1}(j\in S), j=1,\ldots,p$. A subspace following the \textit{hierarchical restrictive multinomial distribution} $\mathcal{R}(\mathcal{U}, p, \tilde{\bm{\eta}})$, where $\sum_{j=1}^p \tilde{\eta}_j = 1$ and $\tilde{\eta}_j \geq 0$, is equivalent to the procedure:
\begin{enumerate}
	\item Draw $d$ from  distribution $\mathcal{U}$ on $\{1,\ldots, D\}$;
	\item Draw $\bm{J}=(J_1,\ldots, J_p)^T$ from a restrictive multinomial distribution with parameter $(p, d, \tilde{\bm{\eta}})$, where the restriction is $J_j \in \{0, 1\}$.
\end{enumerate}
For example, the hierarchical uniform distribution belongs to this family where $\mathcal{U}$ is the uniform distribution $\mathcal{U}_0$ on $\{1,\ldots, D\}$ and $\tilde{\eta}_j = \frac{1}{p}$ for all $j = 1, \ldots, p$.

With the hierarchical restrictive multinomial distribution in hand, we can depict the iterative algorithm more precisely. At iteration $t$, the algorithm updates the subspace distribution of next round $\mathcal{D}^{[t+1]}$ by the hierarchical restrictive multinomial distribution $\mathcal{R}(\mathcal{U}_0, p, \tilde{\bm{\eta}}^{[t]})$, where $\tilde{\eta}_j^{[t]} \propto [\heta_j^{[t]}\mathds{1}(\heta_j^{[t]} > C_0/\log p) + \frac{C_0}{p}\mathds{1}(\heta_j^{[t]} \leq C_0/\log p)]$ and $\heta_j^{[t]}$ is the proportion of variable $j$ in the $B_1$ selected subspaces $\{S_{b_1*}^{[t]}\}_{b_1=1}^{B_1}$.

\subsection{Connections with marginal screening and interaction detection}
Before closing this section and moving into theoretical analysis, we want to point out the connection of RaSE screening approach with the classical marginal screening methods as well as the important problem of interaction detection. 

First of all, it is easy to observe that when $D = 1$ in Algorithm \ref{algo: rase screening}, with proper measure, RaSE screening method reduces to the marginal screening approaches. In this sense, RaSE screening method can be seen as an extension of classical marginal screening frameworks by evaluating subspaces instead of individual predictors. In addition, when there are signals with no marginal contribution, one intuitive idea is to screen all possible interaction terms, which demand extremely high computational costs. For example, screening all the order-$d$ interactions leads to a computational cost of $O(p^d)$. Instead of screening all possible interactions, RaSE randomly chooses some feature subspaces and explores their contributions to the response via a specific criterion. The carefully designed mechanism of generating random subspaces along with the iterative step greatly alleviate the requirement on computation.
 
Second, there has been a great interest in studying screening methods for interaction detection \citep{hao2014interaction, fan2016interaction, kong2017interaction}. The proposed RaSE screening framework works in a different fashion, by evaluating the contribution of variables through the joint contributions in different subspaces. A simulation example (Example \ref{exp_interaction}) where we have 4-way interactions among predictors will be studied to show the effectiveness of RaSE.  


\section{Theoretical Analysis}\label{sec:theory}
In this section, we investigate the theoretical properties of RaSE screening method to help readers understand how it works and why it can succeed in practice. We are not claiming that the assumptions we make are the weakest and conclusions we obtain are the strongest. 

Before moving forward, we first define some notations. For two numbers $a$ and $b$, we denote $a\vee b=\max(a, b)$ and $a\wedge b = \min(a, b)$. For two numerical sequences $\{a_n\}_{n=1}^{\infty}$ and $\{b_n\}_{n=1}^{\infty}$, we denote $a_n = o(b_n)$ or $a_n \ll b_n$ if $\lim\limits_{n \rightarrow \infty}|a_n/b_n| = 0$. Denote $a_n = O(b_n)$ or $a_n \lesssim b_n$ if  $\limsup\limits_{n \rightarrow \infty}|a_n/b_n| < \infty$. When $a_n \lesssim b_n$ and $a_n \gtrsim b_n$ hold at the same time, we write it as $a_n \asymp b_n$. Denote Euclidean norm for a length-$p$ vector $\bx = (x_1, \ldots, x_p)^T$ as $\twonorm{\bx} = \sqrt{\sum_{j=1}^p x_j^2}$. $\bm{1}_p$ represents a length-$p$ vector with all entries 1. For a $p \times p'$  matrix $A = (a_{ij})_{p \times p'}$, define the 1-norm  $\|A\|_{1} = \sup\limits_{j}\sum_{i=1}^p|a_{ij}|$, the operator norm $\|A\|_{2} = \sup\limits_{\bx: \|\bx\| = 1}\|A\bx\|_2$, the infinity norm $\infnorm{A} = \sup\limits_{i}\sum_{j=1}^{p'}|a_{ij}|$ and the maximum norm $\maxnorm{A} = \sup\limits_{i,j}|a_{ij}|$. We also denote the minimal and maximal eigenvalues of a square matrix $A$ as $\lambda_{\min}(A)$ and $\lambda_{\max}(A)$, respectively. Besides, we use different probability notations $\bp, \tp, \p$ to represent probabilities w.r.t. randomness from subspaces, randomness from training samples, and all randomness, respectively. And we use the same fonts $\be, \te, \e$ to represent the corresponding expectations. In addition, throughout this section, we assume $p^* = |S^*|$ is fixed.

\subsection{Sure screening property}\label{subsec:sure screening}
First, note that the success of RaSE relies on the large selected proportions of all signals. According to Algorithm \ref{algo: rase screening}, the selected proportion of signal $j$ depends on the comparison of two different types of subspaces, namely ``covering signal $j$" or ``not covering signal $j$". To understand when RaSE can succeed, we also need to compare subspaces ``covering a subset $\bar{S}_j \ni j$" or ``not covering $\bar{S}_j$", which is essential when signal $j$ has no marginal effect.  Next, we analyze the joint distribution of these two types of subspaces given the number of $B_2$ subspaces covering some $\bar{S}_j \ni j$, in the following useful lemma.

\begin{lemma}\label{lem: ind}
	Let $\{S_{1b_2}\}_{b_2=1}^{B_2} \overset{i.i.d.}{\sim}  \mathcal{R}(\mathcal{U}_0, p, p^{-1}\bm{1}_p)$. For any set $\bar{S}_j \ni j$ with cardinality $|\bar{S}_j| \leq D$, let $p_j = \bp(S_{11} \supseteq \bar{S}_j) = D^{-1}\sum_{d = |\bar{S}_j|}^{D}\frac{\binom{p-|\bar{S}_j|}{d-|\bar{S}_j|}}{\binom{p}{d}}$. Given $N_j \coloneqq \#\{b_2: S_{1b_2} \supseteq \bar{S}_j\} = k$, dividing $\{S_{1b_2}\}_{b_2=1}^{B_2}$ into $\{\sj{1b_2}\}_{b_2=1}^{k}$ and $\{\smj{1b_2}\}_{b_2=1}^{B_2 - k}$, where $\sj{1b_2} \supseteq \bar{S}_j$ and $\smj{1b_2} \not\supseteq \bar{S}_j$, 	
	\begin{enumerate}[(i)]
		\item $\{\sj{1b_2}\}_{b_2=1}^{k}$ independently follow the distribution
		\begin{equation}\label{eq: sj dist}
			\bp(\sj{} = S) = \left[D\cdot p_j\binom{p}{|S|}\right]^{-1}\cdot \mathds{1}(S \supseteq \bar{S}_j);
		\end{equation}
		\item $\{\smj{1b_2}\}_{b_2=1}^{B_2 - k}$ independently follow the distribution
		\begin{equation}\label{eq: smj dist}
			\bp(\smj{} = S) = \left[D(1-p_j)\binom{p}{|S|}\right]^{-1} \cdot\mathds{1}(S \not\supseteq \bar{S}_j);
		\end{equation}
		\item $\{\sj{1b_2}\}_{b_2=1}^{k} \ind \{\smj{1b_2}\}_{b_2=1}^{B_2 - k}$.
	\end{enumerate}
\end{lemma}

The proof of Lemma \ref{lem: ind} can be found in Appendix \ref{sec: proofs}. It shows us that given $N_j \coloneqq \#\{b_2: S_{1b_2} \supseteq \bar{S}_j\} = k$, $\{\sj{1b_2}\}_{b_2=1}^{k}$ and $\{\smj{1b_2}\}_{b_2=1}^{B_2 - k}$ are independent. And each $\sj{1b_2}, \smj{1b_2}$ follows a ``weighted" hierarchical uniform distribution by adjusting the sampling weight based on the cardinality of subspace. 

Now, we introduce a concentration of $\cri_n$ on its population version $\cri$ for a collection of subsets.  In particular, for any $D$, there exists a sequence $\{\epsilon_n\coloneqq\epsilon(n, D)\}_{n=1}^{\infty}$ and positive constant $c_{1n} \rightarrow 0$ such that 
	\begin{equation}\label{eq: large deviation}
		\tp\left(\supp |\cri_n(S) - \cri(S)| > \epsilon_n \right) \leq c_{1n}
	\end{equation}
holds for any $n$. Such a sequence $\{\epsilon_n\}_{n=1}^{\infty}$ always exists, though we would like it to be small to have a uniform concentration as described in the following assumption, which is important to establish the sure screening property of RaSE.
		

\begin{assumption}\label{asmp: sure screening}
	For any $j = 1, \ldots, p$, there exists a subset $\bar{S}_j \ni j$, and we denote $\delta_j(S)\coloneqq \delta_j(n, D, S) =  \bp_{\sj{}}(\cri(S) -  \cri(\sj{}) < 2\epsilon_n|S)$, where $\sj{}$ follows the distribution in \eqref{eq: sj dist} w.r.t. $\bar{S}_j$. It holds that 
	\begin{equation}
		D \geq \sup_{j \in S^*}|\bar{S}_j|, B_2\inf_{j \in S^*}p_j  \gtrsim 1, \limsup_{n, D, B_2 \rightarrow \infty} \left\{B_2\sup\limits_{j \in S^*}\be_{\smj{}}\left[\delta_j(\smj{})^{\frac{1}{2}B_2p_j}\right]\right\} < \infty,
	\end{equation}
	where $\smj{}$ follows the distribution in \eqref{eq: smj dist} and $p_j = \bp(S_{11} \supseteq \bar{S}_j) = D^{-1}\sum\limits_{d = |\bar{S}_j|}^{D}\frac{\binom{p-|\bar{S}_j|}{d-|\bar{S}_j|}}{\binom{p}{d}}$.
\end{assumption}

\begin{remark}\label{rmk: sure screening}
	 In Assumption \ref{asmp: sure screening}, $\delta_j$ measures the strength of signal $j$ via comparing the two types of feature subspaces introduced in Lemma \ref{lem: ind}. From the assumption, we need a large $B_2$ when $\delta_j$ is small. 
\end{remark}

\begin{theorem}[Sure screening property]\label{thm: sure screening}
	Define 
	\begin{equation}
		c_{2n} \coloneqq c_2(n, B_2, D) \coloneqq (1-c_{1n})\left(1-\sup\limits_{j \in S^*}\be_{\smj{}}\left[\delta_j(\smj{})^{\frac{1}{2}B_2p_j}\right]\right)^{B_2}\left(1-\exp\left\{-\frac{3}{28}B_2\inf_{j \in S^*} p_j\right\}\right).
		\end{equation}
	For any $\alpha>1$, let $\hat{S}_{\alpha} = \{1 \leq j \leq p: \heta_j \textup{ is among the } \allowbreak [\alpha D/c_{2n}] \textup{ largest of all}\}$. Under Assumption  \ref{asmp: sure screening}, when $B_1 \gg \log p^*$ and $n \rightarrow \infty$, we have
	\begin{enumerate}[(i)]
	\item 	$\p(S^* \subseteq \hat{S}_{\alpha}) \geq 1-p^*\exp\left\{-2B_1 c^2_{2n}\left(1-\frac{1}{\alpha}\right)^2\right\} \rightarrow 1$;
	\item The selected model size $|\hat{S}_{\alpha}|\lesssim D$. 
	\end{enumerate}
\end{theorem}

Next, we would like to analyze the restriction on $B_2$ imposed by Assumption \ref{asmp: sure screening}, which depends on $\delta_j$. We first introduce a useful notion called \textit{detection complexity}.

\begin{definition}[Detection complexity]\label{def: detection}
	We say feature $j \in S^*$ is detectable in complexity $d$, if there exists a subset $\bar{S}_j \ni j$ with cardinality $d$ and another subset $S_j^0 \subseteq S_{\textup{Full}} \backslash \{j\}$ with cardinality $p_j^0$, such that
	\begin{equation}
		\inf_{S \in \mathscr{S}, S' \in \mathscr{S}'} \left[\cri(S) - \cri(S')\right]> 2\epsilon_n,
	\end{equation}
	where $\mathscr{S} = \mathscr{S}(j, D) = \{S: |S| \leq D, |S\cap (S^* \cup S_j^0)| < d\}$, $\mathscr{S}' = \mathscr{S}'(j, D) = \{S: |S| \leq D, S \supseteq \bar{S}_j\}$, and $\epsilon_n$ satisfies \eqref{eq: large deviation}. We define the detection complexity of $j$, which is denoted by $\mathpzc{d}_j$, as the minimal integer $d$ to make $j \in S^*$ detectable in complexity $d$.
\end{definition}

\begin{remark}
	The detection complexity $\mathpzc{d}_j$ actually indicates the difficulty to identify signal $j$. When $\mathpzc{d}_j = 1$, $\bar{S}_j$ is actually equal to $\{j\}$ and $\mathscr{S} = \{S: |S|\leq D, S \cap(S^* \cup S^0_j) = \emptyset\}$. It implies that the given criterion function performs better at subsets covering $j$ than at subsets not intersecting with $S^* \cup S^0_j$. $S^0_j$ is introduced to avoid cases that some noises might have strong marginal effects. This condition is similar to marginal conditions in literature, for examples, see \cite{fan2008sure, fan2011nonparametric, zhu2011model, li2012feature, shao2014martingale, cui2015model, pan2018generic, nandy2020covariate}. The difference is that here we state it via subspaces instead of single features used in existing works. And when $\mathpzc{d}_j \geq 2$, the definition of detection complexity allows us to consider the joint contribution of multiple features. See Examples \ref{exp_sis} and \ref{exp_isis} in our numerical studies as examples.	
\end{remark}

Now we introduce an assumption under which the restriction on $B_2$ can be explicitly calculated.

\begin{assumption}\label{asmp: sure screening joint contribution}
	All signals in $S^*$ are detectable in complexity $\mathpzc{d}$, where $\mathpzc{d} = \max\limits_{j \in S^*}\mathpzc{d}_j$.
\end{assumption}

Intuitively speaking, this assumption requires all signals to be detectable under the same level, which equals the largest detection complexity of signals. In some sense, it is necessary for the sure screening property. Signal $j$ with large detection complexity is associated with a larger set $\bar{S}_j$, requiring a larger $B_2$ to sample subsets that cover $\bar{S}_j$ with sufficiently high probability. 

\begin{proposition}\label{prop: sure screening joint contribution}
	Under Assumption \ref{asmp: sure screening joint contribution}, when $B_2 \asymp \left(\frac{p}{D}\right)^{\mathpzc{d}}$, Assumption \ref{asmp: sure screening} hold.
\end{proposition}


In the ideal case, in Assumption \ref{asmp: sure screening}, we can set $\bar{S}_j = \{j\}$ for all $j \in S^*$, implying $\mathpzc{d}=1$, which leads to the weakest restriction on $B_2$, that is, $B_2 \asymp \frac{p}{D}$. If a signal $j$ does not have marginal contribution to the response,  we have $\mathpzc{d}_j \geq 2$, requiring a larger order of $B_2$ to satisfy Assumption \ref{asmp: sure screening}. This motivates the iterative RaSE screening (Algorithm \ref{algo: iterative rase screening}) which usually has a less stringent restriction on $B_2$, making the framework more applicable to  high-dimensional settings.  

Next, we study the sure screening property for iterative RaSE, and discuss how the restriction on $B_2$ can be relaxed. For simplicity, we only study the one-step iteration, i.e. the case when $T=1$. It's not very hard to generalize the conditions and conclusions to the general case when $T>1$. To better state the results, we first generalize Lemma \ref{lem: ind} to understand the distribution of two aforementioned types of subspaces after one iteration.

\begin{lemma}\label{lem: ind general}
For any set $\bar{S}_j \ni j$ with cardinality $|\bar{S}_j| \leq D$, let $\{S_{1b_2}\}_{b_2=1}^{B_2} \overset{i.i.d.}{\sim}$ some distribution $\mathcal{F}$ such that $\bp_{S_{11} \sim \mathcal{F}}(S_{11} \supseteq \bar{S}_j)\in(0,1)$. 
	Given $N_j \coloneqq \#\{b_2: S_{1b_2} \supseteq \bar{S}_j\} = k$, dividing $\{S_{1b_2}\}_{b_2=1}^{B_2}$ into $\{\sj{1b_2}\}_{b_2=1}^{k}$ and $\{\smj{1b_2}\}_{b_2=1}^{B_2 - k}$, where $\sj{1b_2} \supseteq \bar{S}_j$ and $\smj{1b_2} \not\supseteq \bar{S}_j$,
	\begin{enumerate}[(i)]
		\item $\{\sj{1b_2}\}_{b_2=1}^{k}$ independently follow the distribution
		\begin{equation}\label{eq: sj dist general}
			\bp(\sj{} = S) = \bp_{S_{11} \sim \mathcal{F}}(S_{11} = S)\cdot \frac{\mathds{1}(S \supseteq \bar{S}_j)}{\bp_{S_{11} \sim \mathcal{F}}(S_{11} \supseteq \bar{S}_j)};
		\end{equation}
		\item $\{\smj{1b_2}\}_{b_2=1}^{B_2 - k}$ independently follow the distribution
		\begin{equation}\label{eq: smj dist general}
			\bp(\smj{} = S) = \bp_{S_{11} \sim \mathcal{F}}(S_{11} = S)\cdot \frac{\mathds{1}(S \not\supseteq \bar{S}_j)}{\bp_{S_{11} \sim \mathcal{F}}(S_{11} \not\supseteq \bar{S}_j)};
		\end{equation}
		\item $\{\sj{1b_2}\}_{b_2=1}^{k} \ind \{\smj{1b_2}\}_{b_2=1}^{B_2 - k}$.
	\end{enumerate}
\end{lemma}
We omit the proof of Lemma \ref{lem: ind general} as it is very similar to Lemma \ref{lem: ind}. Next, we introduce the following technical assumption analogous to Assumption \ref{asmp: sure screening}.
\begin{assumption}\label{asmp: sure screening iterative}
	Suppose the signal set $S^*$ can be decomposed as $S^* = S_{\mcf}^* \cup S_{\mcs}^*$, where $S_{\mcf}^*$ and $S_{\mcs}^*$ satisfy the following conditions:
	\begin{enumerate}[(i)]
		\item (The first-step detection) For any $j \in S_{\mcf}^*$, denote $\delta_j^{\mcf}(S) =  \bp_{\sj{}}(\cri(S) -  \cri(\sj{}) < 2\epsilon_n|S)$, $p^{\mcf} = \bp(j \in S_{11}) = \frac{D+1}{2p}$, where $\sj{}$ follows the distribution in \eqref{eq: sj dist} w.r.t. $\bar{S}_j = \{j\}$. Then,
		\begin{equation}
			B_2  \gtrsim \frac{p}{D}, \limsup_{n, D, B_2 \rightarrow \infty} \left\{B_2 \sup_{j \in S_{\mcf}^*} \be_{\smj{}}\left[\delta^{\mcf}_j(\smj{})^{\frac{1}{2}B_2p^{\mcf}}\right]\right\}< \infty,
		\end{equation}
		where $\smj{}$ follows the distribution in \eqref{eq: smj dist} w.r.t. $\bar{S}_j$.
		\item (The second-step detection) Denote $\delta_j^{\mcs}(S) =  \bp_{\sj{}}(\cri(S) -  \cri(\sj{}) < 2\epsilon_n|S)$, $p^{\mcs} = \frac{(D+1)C_0}{2(D+C_0)p}$, where $\sj{}$ follows the distribution in \eqref{eq: sj dist general} w.r.t. $\bar{S}_j = \{j\}$, and $C_0$ is a constant from Algorithm \ref{algo: iterative rase screening}. $\Upsilon = \{\mathcal{R}(\mathcal{U}_0, p, \tilde{\bm{\eta}})\}$ is a family of hierarchical restrictive multinomial distributions satisfying
		\begin{equation}
			 \inf_{j \in S_{\mcf}^*} \tilde{\eta}_j \geq \frac{c_2^*}{(D+C_0)},
		\end{equation}
		for a constant $c_2^* > 0$. Then,
		\begin{equation}
			B_2 \gtrsim p, \limsup_{n, D, B_2 \rightarrow \infty} \left\{B_2 \sup\limits_{\mathcal{F} \in \Upsilon}\sup\limits_{j \in S^*} \be_{\smj{}}\left[\delta^{\mcs}_j(\smj{})^{\frac{1}{2}B_2p^{\mcs}}\right]\right\}< \infty,
		\end{equation}
		where $\smj{}$ follows the distribution in \eqref{eq: smj dist general} w.r.t. $\bar{S}_j = \{j\}$ and $\{S_{b1b2}\}_{b_1, b_2} \overset{i.i.d.}{\sim} \mathcal{F} \in \Upsilon$. 
	\end{enumerate}
\end{assumption}

\begin{remark}
	Condition (\rom{1}) is a relaxed version of Assumption \ref{asmp: sure screening}, which replaces $S^*$ by a  subset $S_{\mcf}^*$. 
	This can be seen as a first-step detection condition for RaSE screening method to capture $S^*_{\mcf}$. The remaining signals in $S^*_{\mcs}$ that might be missed in the first step  will be captured in the second step. The family of distributions $\Upsilon$ is introduced to incorporate the randomness in the first step of RaSE screening.  This type of stepwise detection condition is very common in the literature \citep{jiang2014variable, li2019robust, zhou2020model, tian2021rase}.
\end{remark}

\begin{theorem}[Sure screening property for one-step iterative RaSE screening]\label{thm: sure screening iterative}
	Define
	\begin{equation}
	c_{2n}^{[l]} \coloneqq c_2^{[l]}(n, B_2, D) \coloneqq (1-c_{1n})\left(1-\sup\limits_{j \in S^*}\be_{\smj{}}\left[\delta_j^{[l]}(\smj{})^{\frac{1}{2}B_2p^{[l]}}\right]\right)^{B_2} \left(1-\exp\left\{-\frac{3}{28}B_2(p^{[l]})^2\right\}\right),
	\end{equation}
	where $l = 0, 1$. For $\hat{S}_{\alpha}^{[1]} = \{1 \leq j \leq p: \heta_j^{[1]} \textup{ is among the } [\alpha D/c_{2n}^{\mcs}] \textup{ largest of all}\}$, where $\alpha > 1$, under Assumption \ref{asmp: sure screening iterative}, if $c_{2n}^{\mcf} > c_2^*$ and $B_1 \gg \log p^*$, we have
\begin{enumerate}[(i)]
	\item $\p(S^* \subseteq \hat{S}_{\alpha}^{[1]}) \geq 1-p^*\exp\left\{-2B_1 (c_{2n}^{\mcf} - c_{2}^{*})^2\right\} -p^*\exp\left\{-2B_1 (c_{2n}^{\mcs})^2\left(1-\frac{1}{\alpha}\right)^2\right\} 
		\rightarrow 1$, 
	as $n \rightarrow \infty$;
	\item $|\hat{S}_{\alpha}^{[1]}| \lesssim D$.
\end{enumerate}
		
\end{theorem}

The lower bound in (\rom{1}) comes from the two steps of Algorithm \ref{algo: iterative rase screening}, which is very intuitive. The general iterative RaSE screening algorithm with any $T \geq 1$ can be studied similarly by imposing analogous conditions, which we leave as future work. 

The restriction on $B_2$ can be discussed in a similar fashion as the vanilla RaSE screening for some specific scenarios. For instance, similar to Definition \ref{def: detection}, we can define the detection complexity of the second step based on the distribution of subsets from the first step. If a similar assumption like Assumption \ref{asmp: sure screening joint contribution} (see Assumption \ref{asmp: iter b2} in Appendix \ref{subsec: rase iteration} for the precise statement) holds, then we can expect that there exist $B_2 \asymp p/D$ in the first step and $B_2 \lesssim D^{|S^*_{[0]}|}(\log p)^{|S^*_{[1]}|-1}p$ in the second step to make Assumption \ref{asmp: sure screening iterative} hold (see Proposition \ref{prop: iter b2} in Appendix \ref{subsec: rase iteration} for a precise description), which relaxes the requirement shown in Proposition \ref{prop: sure screening joint contribution} ($B_2 \asymp (p/D)^{|S^*_{[0]}|+1}$) to a great extent. In Section \ref{sec:numerical}, an array of simulations and real data analyses will show the effectiveness of iterative RaSE screening. 

\subsection{Rank consistency}\label{subsec:rank consistency}
Next, we study another important property of the RaSE screening, namely the rank consistency. First, we impose the following assumption.
\begin{assumption}\label{asmp: rank consistency}
	 Suppose the following conditions hold:
\begin{enumerate}[(i)]
	\item Denote $\tilde{\delta}_j(S) = \bp_{\smj{}}(\cri(S) - 2\epsilon_n < \cri(\smj{})|S)$ and $\delta_j(S)\coloneqq \delta_j(n, D, S) =  \bp_{\sj{}}(\cri(S) -  \cri(\sj{}) < 2\epsilon_n|S)$, where $\smj{}$ follows the distribution in \eqref{eq: smj dist} with respect to $\bar{S}_j = \{j\}$ while $\sj{}$ follows the distribution in \eqref{eq: sj dist} with respect to some subset $\bar{S}_j \ni j$. We have
	\begin{align}
	&\gamma(n, D, B_2) \coloneqq (1-c_{1n})\left(1-2\exp\left\{-\frac{3}{28}B_2\inf_{j \in S^*}p_j\right\}\right) \\
	&\quad \cdot\left[\left(1-\sup_{j \in S^*}\be_{\smj{}}[\delta_{j}(\smj{})^{\frac{1}{2}B_2p_j}]\right)^{B_2}+\left(1-\sup_{j \notin S^*}\be_{\sj{}}\left[\tilde{\delta}_j(\sj{})^{B_2-\frac{3}{2}B_2p^{[0]}}\right]\right)^{\frac{3}{2}B_2p^{[0]}}-1\right] \\
	&\quad - c_{1n}> 0,
	\end{align}
	where $\sj{}$ and $\smj{}$ follow the distributions in \eqref{eq: sj dist} and \eqref{eq: smj dist}, respectively.
	\item $B_1 \gg \gamma(n, D, B_2)^{-2} \vee \log p$.
\end{enumerate}
\end{assumption}
\begin{remark}
	Condition (\rom{1}) is introduced to make sure the signals are separable from the noises. Here, $\tilde \delta_j$ is a parallel definition to $\delta_j$,  measuring the noise level via comparing the two types of feature subspaces introduced in Lemma \ref{lem: ind}. A related condition can be found in Assumption (C3) of \cite{cui2015model}. 
\end{remark}
\begin{theorem}[Rank consistency]\label{thm: rank consistency}
	Under Assumption \ref{asmp: rank consistency}, 
	\begin{equation}
		\p\left(\inf_{j \in S^*} \heta_j > \sup_{j \notin S^*} \heta_j\right) \geq 1-p\exp\left\{-\frac{1}{2}B_1 \gamma^2(n, D, B_2)\right\} \rightarrow 1,
	\end{equation}
	as $n, B_1, B_2 \rightarrow \infty$.
\end{theorem}

In addition, under Assumption \ref{asmp: sure screening joint contribution} with $\mathpzc{d} = 1$, when $B_2$ is restricted to some level, we have Assumption \ref{asmp: rank consistency} holds by default.

\begin{proposition}\label{prop: rank consistency}
	Under Assumption \ref{asmp: sure screening joint contribution} with $\mathpzc{d} = 1$, there exist constants $C_2 > C_1 > 0$, such that, when $B_2 \in (C_1p/D, C_2p/D)$, Assumption \ref{asmp: rank consistency} holds.
\end{proposition}

\section{Numerical Studies}\label{sec:numerical}
In this section, we will investigate the performance of RaSE screening methods via extensive simulations and real data experiments. Each setting is replicated 200 times. In simulations, we evaluate different screening approaches by calculating the 5\%, 25\%, 50\%, 75\%, and 95\% quantiles of the minimum model size (MMS) to include all signals. The smaller the quantile is, the better the screening approach is. For real data, since $S^*$ is unknown, we compare different methods by investigating the performance of the corresponding post-screening procedure. That is, after screening, we keep the same number of variables for each screening method, then the same model is fitted based on those selected variables and their prediction performance on an independent test data is reported. 

We compare RaSE screening methods with SIS \citep{fan2008sure}, ISIS \citep{fan2008sure, fan2009ultrahigh}, SIRS \citep{zhu2011model}, DC-SIS \citep{li2012feature}, MDC-SIS \citep{shao2014martingale}, MV-SIS \citep{cui2015model}, HOLP \citep{wang2016high}, IPDC \citep{kong2017interaction}, and CIS \citep{nandy2020covariate}.

All the experiments are conducted in R. We implement RaSE screening methods in \texttt{RaSEn} package. R package \texttt{SIS} is used to implement SIS. Corresponding to one-step iterative RaSE, we report the results of ISIS with two screening steps and one selection step \citep{saldana2018sis}.\footnote{For more details, please refer to the toy example in Appendix \ref{subsubsec: isis}.} R package \texttt{screening} (\url{https://github.com/wwrechard/screening}) is used to implement HOLP. We conduct SIRS, DC-SIS and MV-SIS through R package \texttt{VariableSelection}. IPDC is implemented by calling the function \texttt{dcor} in R package \texttt{energy}. We implement MDC-SIS through function \texttt{mdd} in R package \texttt{EDMeasure} to calculate the martingale difference divergence. CIS is implemented via R codes shared in \cite{nandy2020covariate}. 

We combine RaSE framework with various criteria to choose subspaces, including minimizing BIC (RaSE-BIC) and eBIC (RaSE-eBIC)  in linear model or logistic regression model, minimizing the leave-one-out MSE/error in $k$-nearest neighbor ($k$NN) (RaSE-$k$NN), and minimizing the 5-fold cross-validation MSE/error in support vector machine (SVM) with RBF kernel (RaSE-SVM). We add a subscript 1 to RaSE to denote the one-step iterative RaSE (e.g. RaSE\textsubscript{1}-BIC). In practice, we can choose the criterion based on the model we prefer in the post-screening procedure. For example, if we would like to use linear model in post-screening, we could set minimizing BIC of linear model as the criterion. If we want to fit a non-linear model in post-screening, minimizing cross-validation error in $k$NN or SVM with RBF kernel can be good choices. Some exploratory analysis can help us choose a proper post-screening method.

For all RaSE methods, we fix $B_1=200$ and $B_2=20\times [p/D]$, motivated by Proposition \ref{prop: sure screening joint contribution}. In addition, following \cite{weng2019regularization}, we fix $D=[\sqrt{n}]$, which is motivated from the fact that many estimators are $\sqrt{n}$-consistent. And we verify the effectiveness of this choice in Example \ref{exp_sis}. For Example \ref{exp_sis}, we also investigate the impact of $B_1$, $B_2$ and $D$ on the median MMS. For RaSE-$k$NN and RaSE\textsubscript{1}-$k$NN, $k$ is set to be 5. For RaSE-eBIC and RaSE\textsubscript{1}-eBIC, we set the penalty parameter $\gamma = 0.5$ \citep{chen2008extended, chen2012extended}. 


All the codes used in numerical experiments can be found on GitHub (\url{https://github.com/ytstat/RaSE-screening-codes}).

\subsection{Simulations}\label{subsec: simulations}

\begin{example}[Example \Rom{2} in \cite{fan2008sure}]\label{exp_sis}
	We generate data from the following model:
	\begin{equation}\label{eq: model in exp 1}
		y = 5x_1 + 5x_5 + 5x_3 - \frac{15}{\sqrt{2}}x_4 + \epsilon,
	\end{equation}
	where $\bx = (x_1, \ldots, x_p)^T \sim N(\bm{0}, \Sigma)$, $\Sigma = (\sigma_{ij})_{p \times p}$, $\sigma_{ij} = 0.5^{\mathds{1}(i \neq j)}$, $\epsilon\sim N(0, 1)$, and $\epsilon \ind \bx$.  The signal set $S^* = \{1, 2, 3, 4\}$. $n=100$ and $p=1000$.
\end{example}

In this example, there is no correlation between $y$ and $x_4$, further leading to the independence due to normality, therefore methods based on the marginal effect will fail to capture $x_4$. However, after projecting $y$ on the space which is perpendicular with any signals from $x_1$, $x_2$ and $x_3$, the correlation appears between the projected $y$ and $x_4$, which motivates the ISIS. Besides, the proposed RaSE methods are also expected to succeed since it works with feature subsets instead of a single variable.

\begin{table}[!h]
\renewcommand{\arraystretch}{0.7}
\setlength{\tabcolsep}{7pt}
\begin{center}
\begin{threeparttable}
\begin{tabular}{l|rrrrr|rrrrr}
\Xhline{1pt}
\multirow{2}{*}{Method/MMS} & \multicolumn{5}{c}{Example 1} & \multicolumn{5}{c}{Example 2} \\ \cline{2-11}
&5\% &25\%  &50\%  &75\%  &95\% &5\% &25\%  &50\%  &75\%  &95\%  \\
\hline
SIS&  227&  317&  397&  647&  922& 6&  28&  105&  592&  1855\\
ISIS&  14&  15&  15&  15&  25& 172&  861&  1415&  1825&  1963\\
SIRS&  87&  370&  594&  762&  949& 6&  1158&  1492&  1774&  1964\\
DC-SIS&  96&  358&  610&  776&  942& 6&  1083&  1460&  1752&  1976\\
HOLP&  912&  949&  969&  986&  999& 45&  196&  576&  1252&  1906\\
IPDC&  224&  442&  700&  869&  980& 59&  210&  386&  678&  1517\\
MDC-SIS&  146&  287&  512&  734&  937& 6&  20&  93&  999&  1908\\
CIS&  203&  434&  601&  780&  940& 2000&  2000&  2000&  2000&  2000\\
RaSE-BIC&  5&  12&  37&  126&  650& 6&  358&  1514&  1821&  1956\\
RaSE\textsubscript{1}-BIC&  4&  4&  4&  16&  55& 13&  834&  1507&  1797&  1969\\
RaSE-eBIC&  6&  21&  42&  489&  852& 8&  26&  1323&  1789&  1935\\
RaSE\textsubscript{1}-eBIC&  4&  4&  4&  4&  14& 907&  1485&  1739&  1878&  1971\\
RaSE-$k$NN&  22&  88&  233&  312&  883& 5&  5&  6&  76&  1190\\
RaSE\textsubscript{1}-$k$NN&  6&  80&  422&  694&  921& 5&  5&  5&  13&  1846\\
RaSE-SVM&  13&  59&  150&  336&  842& 5&  5&  5&  6&  68\\
RaSE\textsubscript{1}-SVM&  4&  4&  82&  126&  542& 5&  5&  5&  5&  11\\
\Xhline{1pt}
\end{tabular}
\end{threeparttable}
\caption{Quantiles of MMS in Examples \ref{exp_sis} and \ref{exp_knn}.}
\label{table: exp sis_knn}
\end{center}
\end{table}

We present the results in the left panel of Table \ref{table: exp sis_knn}. From the results, it can be seen that all the marginal screening methods do not perform well in the sense that they need a large model to cover all 4 signals. ISIS performs much better because it can detect the signals with a smaller model than SIS with one step iteration. For RaSE screening methods with no iteration, as analyzed in Proposition \ref{prop: sure screening joint contribution},  we have $\mathpzc{d} = 2$ since $x_4$ has no marginal contribution to $y$, leading to a theoretical requirement $B_2 \asymp (p/D)^2$, where $(p/D)^2 = 10^4$. Despite the current small $B_2$ setting, RaSE-BIC and RaSE-eBIC still perform better than SIS and other marginal screening methods. After one iteration, RaSE\textsubscript{1}-BIC and RaSE\textsubscript{1}-eBIC improve a lot compared to their vanilla counterparts, with RaSE\textsubscript{1}-eBIC achieving the best performance. 

Note that iterations can usually improve the performance of vanilla RaSE at small quantiles, but possibly lead to worse performance at large quantiles. See RaSE-$k$NN and RaSE\textsubscript{1}-$k$NN for examples. This phenomenon happens because iterative RaSE is very aggressive and the success of the second step is based on the accurate capture of some signals in the first step. If the first step fails to identify enough signals but captures many noises, these noises will be selected more frequently in the second step.

To further study the impact of $(B_1, B_2)$, we run this example for 200 times under different $(B_1, B_2)$ settings, where we range $B_1$ from 100 to 1000 with increment 100 and $B_2$ from 1000 to 97000 with increment 6000. The median of MMS with RaSE-BIC and RaSE\textsubscript{1}-BIC is summarized in Figure \ref{fig: B1B2_settings}. It shows that in general, larger $(B_1, B_2)$ leads to better performance. The performance is stable in terms of $B_1$ when $B_2$ is large.  On the other hand,  the performance improves continuously as $B_2$ grows. In particular, for RaSE-BIC, when $B_2 \geq 10^4$, it can capture $S^*$ very well, which agrees with Proposition \ref{prop: sure screening joint contribution}. These results indicate that we can further improve the performance of RaSE screening if we have sufficient computational resources. RaSE\textsubscript{1}-BIC can always achieve a great performance with a small $B_2$, showing its effectiveness in relaxing the restriction on $B_2$.

\begin{figure}[!h]
  \centering
  \subfloat[RaSE-BIC]{\label{fig: exp1_b1b2}\includegraphics[width=.5\textwidth]{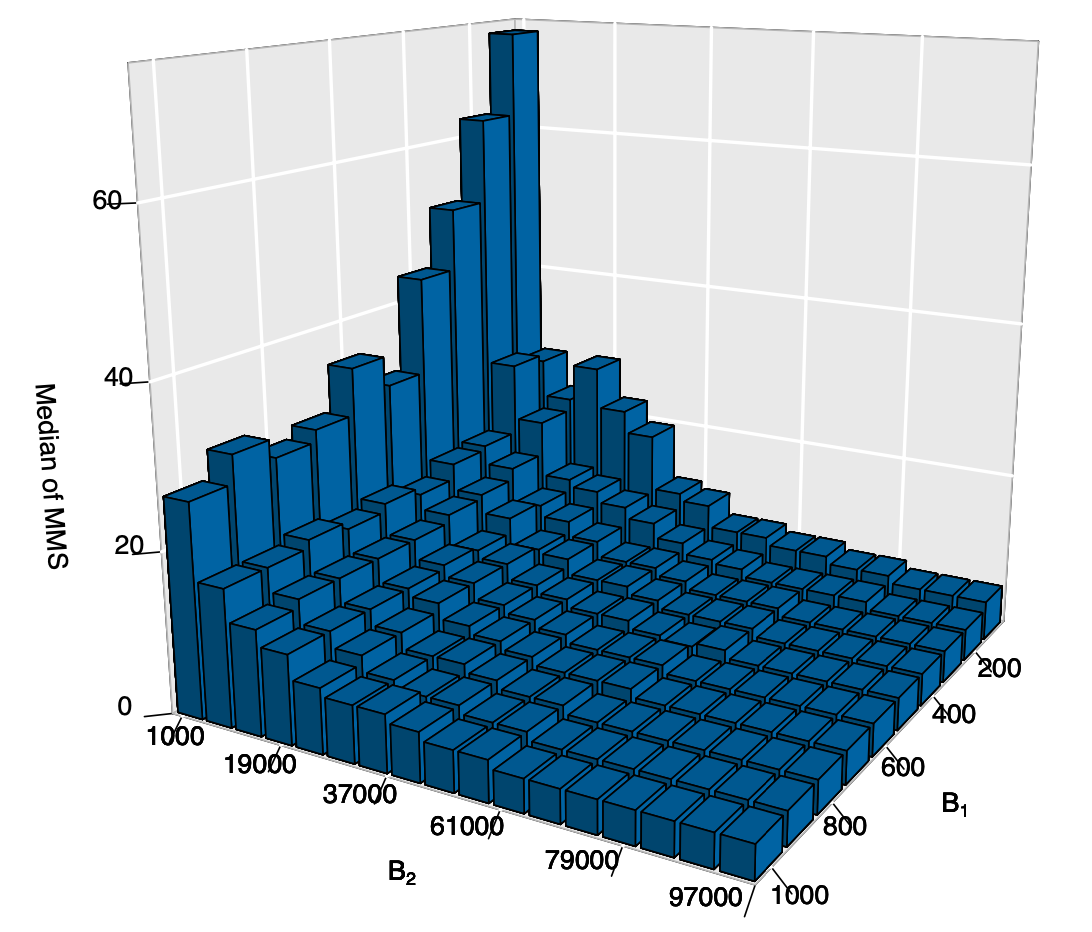}}
  \subfloat[RaSE\textsubscript{1}-BIC]{\label{fig: exp1_b1b2_1}\includegraphics[width=.5\textwidth]{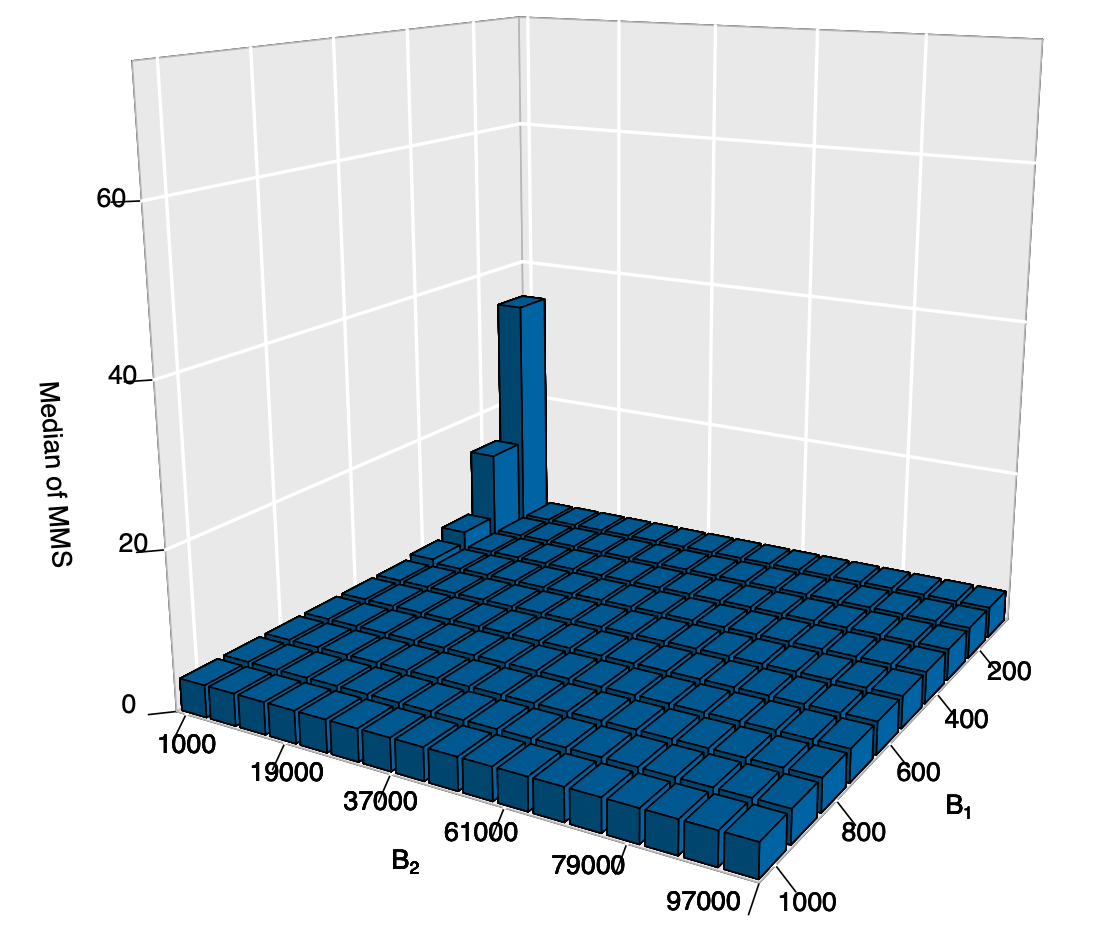}}
  \caption{Median MMS to capture $S^*$ ($|S^*| = 4$) as $(B_1, B_2)$ varies for RaSE-BIC (a) and RaSE\textsubscript{1}-BIC (b) in Example \ref{exp_sis}.}
  \label{fig: B1B2_settings}
\end{figure}

We also run this example 200 times to plot the median of MMS for RaSE-BIC and RaSE\textsubscript{1}-BIC under different $(D, B_2)$ while fixing $B_1=200$ in Figure \ref{fig: D_settings} in Appendix \ref{subsec: rase numerical}, where $D$ ranges from 2 to 40 with increment 2 and $B_2$ from 200 to 5000 with increment 300. The subfigure (a) shows that for RaSE-BIC, for a given $B_2$, the impact of $D$ is not monotonic. RaSE-BIC has a good and stable performance when $D$ is around $\sqrt{n} = 10$, which verifies the effectiveness of our choice for $D$. The subfigure (b) shows that the performance of RaSE\textsubscript{1}-BIC is very robust with respect to $D$, as long as $D$ and $B_2$ are not very small.

To compare the computational time of different methods, we list the average running time in 200 replications of Example \ref{exp_sis} in Table \ref{table: time}. All codes were run on NYU Greene clusters (2x Intel Xeon Platinum 8268 24C 205W 2.9GHz Processor) with 40 cores and 50 GB memory.\footnote{For SIS, ISIS, SIRS, DC-SIS and HOLP, since the package implementing them does not provide the option to use multi-cores, we ran them with a single core only.} It can be seen that RaSE methods have heavier computational burdens than other screening methods since their success leverages generating a large number of subspaces. This  can be alleviated with parallel computing and more powerful machines.

\begin{table}[ht]
\renewcommand{\arraystretch}{0.75}
\setlength{\tabcolsep}{6pt}
\begin{center}
\begin{tabular}{l|ccccccccc}
\Xhline{1pt}
Other methods &SIS&ISIS&SIRS&DC-SIS&HOLP&IPDC&MDC-SIS&CIS\\
\hline
Time (s) &  0.01&  0.67&  0.28&  1.30&  0.02&  0.49&  0.28&  1.18\\
\hline
RaSE methods &BIC& BIC\textsubscript{1} &eBIC &eBIC\textsubscript{1} &$k$NN &$k$NN\textsubscript{1} &SVM &SVM\textsubscript{1}\\
\hline
Time (s) &  1.99&  4.03&  2.01&  3.94&  6.74&  13.66&  150.41&  305.77 \\
\Xhline{1pt}
\end{tabular}
\caption{Average (over 200 replications) computational time in seconds for various methods in Example \ref{exp_sis}. For simplicity, for RaSE methods, we use criteria to differentiate them and the subscript ``1" denotes the one-step iterative version of the corresponding RaSE-based methods.}
\label{table: time}
\end{center}
\end{table}

\begin{example}[Latent clusters]\label{exp_knn}
	We generate data from the following linear model:
	\begin{equation}\label{eq: model in exp 3}
		y = 0.5(\tilde{x}_1 + \tilde{x}_2 + \tilde{x}_3 + \tilde{x}_4 + \tilde{x}_5 + \epsilon),
	\end{equation}
	where $\tilde{\bx} = (\tilde{x}_1, \ldots, \tilde{x}_p)^T \sim N(\bm{0}, \Sigma)$, $\epsilon \sim t_2$, $\Sigma = (\sigma_{ij})_{p \times p} = (0.5^{|i-j|})_{p \times p}$, and $\epsilon \ind \bx$. Generate $z \sim \textup{Unif}(\{-3, 3\}) \ind \tilde{\bx}$ and $\bx = \tilde{\bx} + z\mathbf{1}_p$. The signal set $S^* = \{1, 2, 3, 4, 5\}$.  $n=200$ and $p=2000$.
\end{example}

\begin{figure}[!h]
	\centering
	\includegraphics[width=0.8\textwidth]{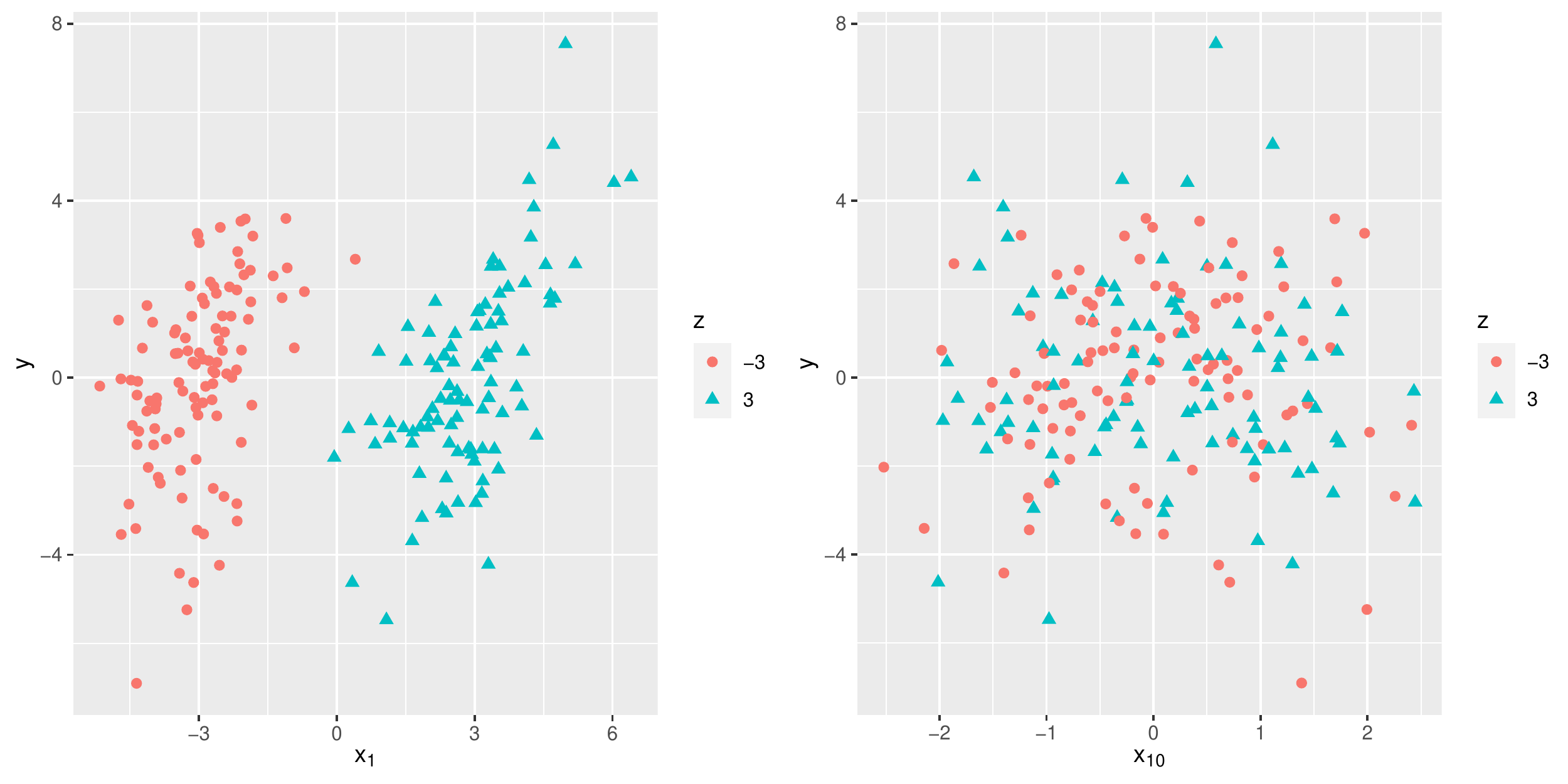}
	\caption{Scatterplots of $y$ vs. $x_1$ and $y$ vs. $x_{10}$ for Example \ref{exp_knn} ($n = 200$).}
	\label{fig: exp_knn}
\end{figure}

Figure \ref{fig: exp_knn} shows the scatterplots of $y$ vs. $x_1$ (left panel) and $y$ vs. $x_{10}$ (right panel). We expect the methods based on Pearson correlation to deteriorate due to the partial cancellation of signals by the averaging of two clusters. For such kind of data, $k$NN could be a favorable approach. The performances of various methods are presented in the right panel of Table \ref{table: exp sis_knn}. SIS and MDC-SIS perform well at $5\%$ and $25\%$ quantiles. RaSE-$k$NN and RaSE-SVM perform quite well with their performances further improved by their respective one-step iterative versions.

\begin{example}[Example 1.c in \cite{li2012feature}]\label{exp_dcsis}
	We generate data from the following model:
	\begin{equation}
		y = 2\beta_1 x_1x_2 + 3\beta_2\mathds{1}(x_{12} < 0)x_{22} + \epsilon,
	\end{equation}
	where $\beta_j = (-1)^U(4\log n/\sqrt{n} + |Z|), j = 1, 2$, $U \sim \textup{Bernoulli}(0.4)$, $Z \sim N(0, 1)$, $\epsilon \sim N(0, 1)$, $\bx \sim N(\bm{0}, \Sigma)$  where $\Sigma = (\sigma_{ij})_{p \times p} = (0.8^{|i-j|})_{p \times p}$,  $U\ind Z$, $\epsilon \ind \bx$,  and $(U, Z)\ind (\epsilon,\bx)$. Note that we regenerate $(U, Z)$ for each replication, so the results might differ from those in \cite{li2012feature}. The signal set $S^* = \{1, 2, 12, 22\}$. $n=200$ and $p=2000$.
\end{example}

\begin{table}[!h]
\renewcommand{\arraystretch}{0.7}
\setlength{\tabcolsep}{7pt}
\begin{center}
\begin{threeparttable}
\begin{tabular}{l|rrrrr|rrrrr}
\Xhline{1pt}
\multirow{2}{*}{Method/MMS} & \multicolumn{5}{c}{Example 3} & \multicolumn{5}{c}{Example 4} \\ \cline{2-11}
&5\% &25\%  &50\%  &75\%  &95\% &5\% &25\%  &50\%  &75\%  &95\%  \\
\hline
SIS&  184&  810&  1370&  1732&  1957& 264&  570&  709&  885&  984\\
ISIS&  362&  1008&  1482&  1775&  1945& 293&  626&  810&  911&  978\\
SIRS&  54&  741&  1294&  1634&  1920& 487&  737&  867&  935&  992\\
DC-SIS&  25&  456&  1222&  1638&  1923& 44&  304&  603&  814&  949\\
HOLP&  326&  954&  1475&  1774&  1975& 316&  586&  767&  886&  974\\
IPDC&  128&  429&  920&  1397&  1899& 7&  19&  68&  158&  528\\
MDC-SIS&  52&  165&  504&  1331&  1872& 189&  482&  736&  889&  979\\
CIS&  4&  5&  8&  55&  548& 5&  33&  136&  352&  789\\
RaSE-BIC&  637&  1242&  1619&  1842&  1959& 355&  693&  825&  914&  986\\
RaSE\textsubscript{1}-BIC&  714&  1196&  1550&  1839&  1974& 424&  661&  824&  918&  981\\
RaSE-eBIC&  484&  1137&  1496&  1794&  1951& 302&  553&  784&  913&  987\\
RaSE\textsubscript{1}-eBIC&  725&  1330&  1617&  1806&  1948& 480&  686&  860&  930&  986\\
RaSE-$k$NN&  5&  33&  168&  1321&  1855& 5&  15&  68&  290&  889\\
RaSE\textsubscript{1}-$k$NN&  4&  5&  8&  125&  1528& 4&  8&  51&  446&  910\\
RaSE-SVM&  4&  18&  504&  1282&  1848& 4&  15&  132&  468&  938\\
RaSE\textsubscript{1}-SVM&  4&  4&  5&  14&  1141& 4&  30&  232&  645&  898\\
\Xhline{1pt}
\end{tabular}
\end{threeparttable}
\caption{Quantiles of MMS in Examples \ref{exp_dcsis} and \ref{exp_interaction}.}
\label{table: exp dcsis_interaction}
\end{center}
\end{table}

The left panel of Table \ref{table: exp dcsis_interaction} exhibits the results of different screening methods. Due to the interaction term and indicator function, approaches based on linear models like SIS, ISIS, HOLP, and RaSE with BIC and eBIC do not perform very well. CIS and RaSE\textsubscript{1}-$k$NN achieve a very good performance at $5\%$, $25\%$ and $50\%$ quantiles. RaSE-$k$NN performs well at $5\%$ and $25\%$ quantiles but worse at others. RaSE-SVM performs well at the first two quantiles. The iteration step improves the performances of RaSE-$k$NN and RaSE-SVM significantly, and RaSE\textsubscript{1}-SVM outperforms all the other methods except at $95\%$ quantile.

\begin{example}[Interactions]\label{exp_interaction}
	We generate data from the following model:
	\begin{equation}
		y = 3\sqrt{|x_1|} + 2\sqrt{|x_1|}x_2^2 + 4\sin(x_1)\sin(x_2)\sin^2(x_3) + 12\sin(x_1)|x_2|\sin(x_3)x_4^2 + 0.5\epsilon,
	\end{equation}
	where $x_1, \ldots, x_p \stackrel{i.i.d.}{\sim} N(0,1)$, $\epsilon \sim N(0, 1)$, and $\epsilon \ind \bx$. The signal set $S^* = \{1, 2, 3, 4\}$. $n=300$ and $p=1000$.
\end{example}

This example evaluates the capability of different screening methods in terms of selecting high-order interactions. The results are summarized in the right panel of Table \ref{table: exp dcsis_interaction}. It can be observed that \textup{RaSE}-$k$NN, \textup{RaSE}\textsubscript{1}-$k$NN, \textup{RaSE}-SVM, \textup{RaSE}\textsubscript{1}-SVM, IPDC, and CIS achieve an acceptable performance, particularly for the lower quantiles. IPDC and CIS perform better at 75\% and 95\% quantiles than all RaSE methods but worse at the other three quantiles than \textup{RaSE}\textsubscript{1}-$k$NN. The remaining methods do not perform well on any of the 5 quantiles. It shows that RaSE framework equipped with minimizing cross-validation MSE on $k$NN or kernel SVM is promising to capture high-order interactions.

\begin{example}[Gaussian mixture, Example 1 in \cite{cannings2017random}]\label{exp_samworth}
	We generate data from the following model:
	\begin{align}
		y \sim \textup{Bernoulli}(0.5), \mbox{    }\bx|y = r \sim \frac{1}{2}N(\bmu_r, \Sigma) + \frac{1}{2}N(-\bmu_r, \Sigma), r = 0, 1,
	\end{align}
	where $\bmu_0 = (2, -2, 0, \ldots, 0)^T$, $\bmu_1 = (2, 2, 0, \ldots, 0)^T$, $\Sigma$ is an identity matrix. The signal set $S^* = \{1, 2\}$. $n=200$ and $p=2000$.
\end{example}

From the scatterplots in Figure \ref{fig: exp_samworth} in Appendix \ref{subsec: rase numerical}, the marginal screening methods are expected to fail because all signals are marginally independent with $y$. The only way to capture the signals is to measure the joint contribution of $(x_1, x_2)$. We summarize the results  in the left panel of Table \ref{table: exp samworth_isis}. 

The table shows that the marginal methods fail as we expected. RaSE with BIC and eBIC fail as well because the data points from the two classes are not linearly separable (Figure \ref{fig: exp_samworth}). SIRS, RaSE\textsubscript{1}-$k$NN and RaSE\textsubscript{1}-SVM achieve the best performance with very accurate feature ranking. 

\begin{table}[!h]
\renewcommand{\arraystretch}{0.7}
\setlength{\tabcolsep}{7pt}
\begin{center}
\begin{threeparttable}
\begin{tabular}{l|rrrrr|rrrrr}
\Xhline{1pt}
\multirow{2}{*}{Method/MMS} & \multicolumn{5}{c}{Example 5} & \multicolumn{5}{c}{Example 6} \\ \cline{2-11}
&5\% &25\%  &50\%  &75\%  &95\% &5\% &25\%  &50\%  &75\%  &95\%  \\
\hline
SIS&  515&  1090&  1414&  1746&  1947& 170&  471&  910&  1436&  1932\\
ISIS&  445&  1001&  1470&  1784&  1967& 7&  7&  7&  8&  8\\
SIRS&  2&  2&  2&  2&  2& 821&  1242&  1551&  1813&  1966\\
DC-SIS&  451&  960&  1385&  1706&  1913& 765&  1155&  1526&  1775&  1947\\
MV-SIS&  379&  957&  1366&  1692&  1895& 199&  706&  1258&  1660&  1909\\
HOLP&  495&  1065&  1381&  1712&  1936& ---&  ---&  ---&  ---&  ---\\
IPDC&  495&  1010&  1344&  1673&  1908& 879&  1425&  1722&  1884&  1988\\
MDC-SIS&  462&  1038&  1332&  1708&  1948& 163&  498&  1064&  1628&  1917\\
CIS&  2000&  2000&  2000&  2000&  2000& 229&  736&  1195&  1652&  1941\\
RaSE-BIC&  506&  1081&  1487&  1804&  1946& 8&  14&  20&  26&  1525\\
RaSE\textsubscript{1}-BIC&  464&  968&  1360&  1692&  1927& 5&  5&  5&  6&  14\\
RaSE-eBIC&  425&  1045&  1424&  1705&  1965& 26&  346&  894&  1406&  1919\\
RaSE\textsubscript{1}-eBIC&  480&  988&  1370&  1727&  1938& 5&  7&  10&  14&  1184\\
RaSE-$k$NN&  2&  3&  5&  6&  8& 38&  202&  294&  1470&  1925\\
RaSE\textsubscript{1}-$k$NN&  2&  2&  2&  2&  2& 27&  376&  967&  1486&  1828\\
RaSE-SVM&  2&  4&  6&  8&  26& 11&  39&  118&  343&  1743\\
RaSE\textsubscript{1}-SVM&  2&  2&  2&  2&  2& 5&  5&  118&  1133&  1792\\
\Xhline{1pt}
\end{tabular}
\end{threeparttable}
\caption{Quantiles of MMS in Examples \ref{exp_samworth} and \ref{exp_isis}.}
\label{table: exp samworth_isis}
\end{center}
\end{table}

\begin{example}[Multinomial logistic regression, Case 2 in Section 4.5 of \cite{fan2009ultrahigh}]\label{exp_isis}
	We first generate $\tilde{x}_1, \ldots, \tilde{x}_4 \overset{i.i.d.}{\sim} \textup{Unif}([-\sqrt{3}, \sqrt{3}])$ and $\tilde{x}_5, \ldots, \tilde{x}_p \overset{i.i.d.}{\sim} N(0, 1)$, then let $x_1 = \tilde{x}_1 - \sqrt{2} \tilde{x}_5$, $x_2 = \tilde{x}_2 + \sqrt{2} \tilde{x}_5$, $x_3 = \tilde{x}_3 - \sqrt{2} \tilde{x}_5$, $x_4 = \tilde{x}_4 + \sqrt{2} \tilde{x}_5$ and $x_j = \tilde{x}_j$ for $j = 5, \ldots, p$. The response is generated from
	\begin{equation}
		 \p(y = r|\tilde{\bx}) \propto \exp\{f_r(\tilde{\bx})\}, r = 1, \ldots, 4,
	\end{equation}
	where $f_1(\tilde{\bx}) = -a\tilde{x}_1 + a\tilde{x}_4$, $f_2(\tilde{\bx}) = a\tilde{x}_1 - a\tilde{x}_2$, $f_3(\tilde{\bx}) = a\tilde{x}_2 - a\tilde{x}_3$ and $f_4(\tilde{\bx}) = a\tilde{x}_3 - a\tilde{x}_4$ with $a = 5/\sqrt{3}$. The signal set $S^* = \{1, 2, 3, 4, 5\}$. $n=200$ and $p=2000$.
\end{example}

In this example, $x_5$ is marginally independent of $y$, therefore the marginal methods are expected to fail to capture $x_5$. Results are summarized in the right panel of Table \ref{table: exp samworth_isis}.

We observe that ISIS, RaSE\textsubscript{1}-BIC, and, RaSE\textsubscript{1}-eBIC lead to better performances. Without iteration, RaSE-BIC still performs competitively compared to other non-iterative approaches. Similar to Example \ref{exp_sis}, the iteration usually improves the performance of vanilla RaSE at small quantiles, but leads to worse performance at large quantiles possibly due to the aggressiveness of iterative RaSE.

To justify the effectiveness of RaSE methods in dealing with more complicated predictors, we add two additional examples in Appendix \ref{subsubsec: additional simulations}. In Example \ref{exp_realistic}, we consider realistic predictors, with the same conditional model $y|\bm{x}$ as in Example \ref{exp_sis}. In Example \ref{exp_mixed}, we use a mix of continuous and discrete variables, with the same conditional model as in Example \ref{exp_knn}. While we have similar findings as in Examples \ref{exp_sis} and \ref{exp_knn}, the performance of most approaches become slightly worse, showing the challenges for analyzing real data.

\subsection{Real data experiments}\label{subsec: real data}

In this section, we investigate the performance of RaSE screening methods on two real data sets. Each data set is randomly divided into training data and test data. As suggested by \cite{fan2008sure}, we select variables via different screening methods on training data, then the LASSO, $k$NN and SVM are fitted based on the selected variables on training data, and finally we evaluate different screening methods based on their corresponding post-screening performance on test data. As benchmarks, we also fit LASSO, $k$NN and SVM models on the training data without screening. Following \cite{fan2008sure}, we choose the top $[n/\log n]$ variables for all screening methods, i.e. let $N = [n/\log n]$ in Algorithms \ref{algo: rase screening} and \ref{algo: iterative rase screening}. Note that we could also choose $[\alpha D]$ variables for any $\alpha > 1$, which is motivated by \eqref{S alpha hat}. Another possibility is to use data-driven strategies. For instance, we could sample out a separate validation data set and use the post-screening validation MSE/classification error to determine $N$.\footnote{For RaSE methods, sometimes there might be less than $[n/\log n]$ variables which have positive selected proportion. In this case, we randomly choose from the remaining variables with 0 selected proportions to have the desired number of selected variables. } We randomly divide the whole data set into 90\% training data and 10\% test data in each of 200 replications, and apply various screening methods on training data. Each time, both training and test data are standardized by using the center and scale of training data. 

\subsubsection{Colon cancer data set}
This data set was collected by \cite{alon1999broad} and consists of 2000 genes measured on 62 patients, of which 40 are diagnosed with colon cancer (class 1) and 22 are healthy (class 0). The information on each gene is represented as a continuous variable. The prediction results are summarized in the left panel of Table \ref{table: real data}. 

The table shows that SIS, ISIS, MDC-SIS, CIS, RaSE-BIC, RaSE\textsubscript{1}-BIC, RaSE-eBIC, and RaSE\textsubscript{1}-eBIC improve the performance of vanilla LASSO. In addition,  RaSE-BIC with LASSO achieves the best performance among all post-screening procedures based on LASSO. Besides, RaSE-$k$NN with $k$NN and RaSE\textsubscript{1}-$k$NN with $k$NN lead to better results than those of vanilla $k$NN. RaSE-SVM and RaSE\textsubscript{1}-SVM also improve the performance of vanilla SVM, demonstrating the effectiveness of RaSE to improve various vanilla methods.

For results of RaSE methods, we also gather the top 10 selected features in 200 replications and calculate the percentages of selection of these top features out of 200 replications. The 10 features with the highest percentages (selection rates) are plotted in Figure \ref{fig: alon_perc} in Appendix \ref{subsec: rase numerical}. We notice that the first few features have high or moderately high selection rates ($100\%$ or $> 50\%$, respectively), implying that they are frequently selected in different replications. These results demonstrate that RaSE-based variable screening methods are reasonably stable.

\begin{table}[!h]
\renewcommand{\arraystretch}{0.7}
\setlength{\tabcolsep}{9pt}
\begin{center}
\begin{tabular}{l|c|ccccc}
\Xhline{1pt}
Screening &Post-screening &Cancer &Eye \\
\hline
--- &\multirow{12}{*}{LASSO} & 0.1792(0.1427) & 0.0103(0.0091) \\
SIS & & 0.1633(0.1407) & \textit{0.0091(0.0068)} \\
ISIS & & 0.1767(0.1444) & \textit{0.0091(0.0068)} \\
SIRS & & 0.2800(0.1734) & 0.0132(0.0123) \\
DC-SIS & & 0.3000(0.1998) & 0.0124(0.0118) \\
MV-SIS & & 0.2958(0.1826) & --- \\
HOLP & & 0.1825(0.1491) & 0.0228(0.0269) \\
IPDC & & 0.1917(0.1464) & 0.0129(0.0132) \\
MDC-SIS & & 0.1600(0.1406) & 0.0103(0.0071) \\
CIS & & 0.1550(0.1332) & 0.0194(0.0231) \\
RaSE-BIC & & \textbf{0.1192(0.1277)} & \textbf{0.0090(0.0066)} \\
RaSE\textsubscript{1}-BIC & & \textit{0.1417(0.1324)} & 0.0123(0.0104) \\
RaSE-eBIC & & 0.3083(0.2118) & 0.0092(0.0069) \\
RaSE\textsubscript{1}-eBIC & & 0.1458(0.1397) & 0.0122(0.0098) \\
\hline
--- &\multirow{3}{*}{$k$NN} & 0.2258(0.1653) & 0.0166(0.0206)  \\
RaSE-$k$NN & & 0.1533(0.1340) & 0.0131(0.0158) \\
RaSE\textsubscript{1}-$k$NN & & 0.1867(0.1500) & 0.0133(0.0161) \\
\hline
--- &\multirow{3}{*}{SVM} & 0.2025(0.1503) & 0.0160(0.0243) \\
RaSE-SVM & & \textit{0.1375(0.1277)} & 0.0158(0.0231) \\
RaSE\textsubscript{1}-SVM & & 0.1858(0.1477) & 0.0158(0.0232) \\
\Xhline{1pt}
\end{tabular}
\caption{Average test classification error rate with standard deviations (in parentheses) for colon cancer data set and average test mean square errors (MSEs) with standard deviations  (in parentheses) for rat eye expression data set. We boldface the values corresponding to the best performances and italicize the values corresponding to the subsequent two best performances. }
\label{table: real data}
\end{center}
\end{table}

\subsubsection{Rat eye expression data set}\label{subsubsec: rat}
This data set was used by \cite{scheetz2006regulation, fan2011nonparametric, wang2016high, zhong2015iterative, nandy2020covariate} among others. It contains the gene expression values corresponding to  18976 probes from the eyes of 120 twelve-week-old male F2 rats. Among the 18976 probes, TRIM32 is the response, which is responsible to cause Bardet-Biedl syndrome. We follow \cite{wang2016high} to focus on the top 5000 genes with the highest sample variance. Therefore, the final sample size is 120 and there are 5000 predictors. The right panel of Table \ref{table: real data} shows the test average mean squared error (MSE) coupled with the standard deviation for each post-screening procedure.

The results show that SIS, ISIS, RaSE-BIC and RaSE-eBIC with LASSO achieve comparable performance, which are better than that of the vanilla LASSO. RaSE-$k$NN with $k$NN and RaSE\textsubscript{1}-$k$NN with $k$NN enhance the vanilla $k$NN method as well. RaSE-SVM with SVM and RaSE\textsubscript{1}-SVM with SVM only slightly improve the vanilla SVM for this data set. Note that MV-SIS is not directly applicable to this data set. It is possible to discretize all the variables to make MV-SIS work. See Section 4.2 in \cite{cui2015model} for details.

Similar to the colon data set, we also demonstrate the stability of RaSE methods in Figure \ref{fig: rat_perc}.
\section{Discussion}\label{sec:discussion}
In this article, we propose a very general screening framework named RaSE screening, based on the random subspace ensemble method. We can equip it with any criterion function for comparing subspaces. By comparing subspaces instead of single predictors, RaSE screening can capture signals without marginal effects on response. Besides, an iterative version of the RaSE screening framework is introduced to enhance the performance of vanilla RaSE and relax the restriction on $B_2$. In the theoretical analysis, we establish sure screening property for both vanilla and iterative RaSE frameworks under some general conditions. The rank consistency is also proved for the vanilla RaSE. We investigate the relationship between the signal strength and the appropriate choice of $B_2$, which shows that in some sense the weaker the signal is, a larger $B_2$ is necessary for RaSE to succeed. In the numerical studies, the effectiveness of RaSE and its iterative version is verified through multiple simulation examples and real data analyses.

The success of RaSE leverages on proper choices of $\cri$ (the criterion), $B_1$ (the number of subspace groups),  $B_2$ (the number of subspace candidates in each group), and $D$ (the maximum subspace size). While we have studied their impacts on the performance of RaSE, there exists potential improvement for choosing these ``tuning" parameters. For example, the subspace distribution at each iteration step could be further generalized, e.g., we can choose $D$ from the empirical distribution of the sizes of the selected $B_1$ subspaces.

There are many other interesting problems worth further studying. The first question is that whether there is an adaptive way to automatically select the number of iterations ($T$). A possible solution is cross-validation and to stop the iteration process when the performance of RaSE on validation data stops improving further. Another interesting topic is to use different $B_2$ values in different iteration steps, which might further speed up the computation. 

\spacingset{1} 

\subsection*{Acknowledgements}
We are grateful to the editor, the AE, and anonymous reviewers for their insightful comments which have greatly improved the scope and quality of the manuscript.
\bibliography{reference.bib}
\bibliographystyle{apalike}
\spacingset{1.5} 
\newpage
\setcounter{page}{1}
\begin{appendices}
\renewcommand{\theequation}{\thesection.\arabic{equation}}
\bigskip
\begin{center}
{\large\bf Supplementary Materials of ``RaSE: A Variable Screening Framework via Random Subspace Ensembles"}
\end{center}

%
%
%
%

\section{Additional Details of This Paper}\label{sec: additional details}
\subsection{Vanilla RaSE algorithm in \cite{tian2021rase}}\label{subsec: rase algo}
See Algorithm \ref{algo: rase}.
\begin{algorithm}
\caption{Random subspace ensemble classification (RaSE)}
\label{algo: rase}
\KwIn{training data $\{(\bm{x}_i, y_i)\}_{i = 1}^n$, new data $\bm{x}$, subspace distribution $\mathcal{D}$, criterion $\mathcal{C}$, integers $B_1$ and $B_2$, type of base classifier $\mathcal{T}$}
\KwOut{predicted label $C^{RaSE}_n(\bx)$, the selected proportion of each feature $\hat{\bm{\eta}}$}
Independently generate random subspaces $S_{b_1b_2} \sim \mathcal{D}, 1 \leq b_1 \leq B_1, 1 \leq b_2 \leq B_2$\\
\For{$b_1 \leftarrow 1$ \KwTo $B_1$}{
  Select the optimal subspace $S_{b_1*}$ from $\{S_{b_1b_2}\}_{b_2 = 1}^{B_2}$ according to $\mathcal{C}$ and $\mathcal{T}$ \\
  Train $C_n^{S_{b_1*}-\ty}$ in subspace $S_{b_1*}$
}
Construct the ensemble decision function $\nu_n(\bm{x}) = \frac{1}{B_1}\sum_{b_1 = 1}^{B_1}C_n^{S_{b_1*}-\ty}(\bm{x})$ \\
Set the threshold $\hat{\alpha}$ according to minimize training error\\
Output the predicted label $C^{RaSE}_n(\bx) = \mathds{1}(\nu_n(\bm{x}) > \hat{\alpha})$, the selected proportion of each feature $\hat{\bm{\eta}} =(\heta_1,\ldots,\heta_p)^T$ where $\heta_l=B_1^{-1}\sum_{b_1=1}^{B_1}\mathds{1}(l\in S_{b_1*}), l=1,\ldots,p$
\end{algorithm}

\subsection{Additional analysis on iterative RaSE}\label{subsec: rase iteration}
\begin{assumption}\label{asmp: iter b2}
	We assume the following conditions hold:
	\begin{enumerate}[(i)]
		\item For any $j \in S^*_{[0]}$ and $d \geq 1$, we have
			\begin{equation}
				\cri(S) - \cri(S') > 2\epsilon_n,
			\end{equation}
			for any $S$ and $S'$ satisfying $|S \cap S^*| < d$, $|S' \cap S^*| \geq d$, $S' \ni j$, and $S \not\ni j$.
		\item For any $j \in S^*_{[1]}$ and $d \geq |S^*_{[0]}|+1$, we have
			\begin{equation}
				\cri(S) - \cri(S') > 2\epsilon_n,
			\end{equation}
			for any $S$ and $S'$ satisfying $|S \cap S^*| < d$, $|S' \cap S^*| \geq d$, $S' \ni j$, and $S \not\ni j$.
	\end{enumerate}
\end{assumption}

\begin{remark}
	By using the definition of detection complexity, Assumption \ref{asmp: iter b2} can be stated as follows.
	\begin{enumerate}[(i)]
		\item Any $j \in S^*_{[0]}$ is detectable in any complexity $d \geq 1$ w.r.t. any $\bar{S}_j \ni j$ and $\bar{S}_j^0 = \emptyset$.
		\item Any $j \in S^*_{[1]}$ is detectable in any complexity $d \geq |S^*_{[0]}|+1$ w.r.t. any $\bar{S}_j \ni j$ and $\bar{S}_j^0 = \emptyset$.
	\end{enumerate}
\end{remark}

\begin{proposition}\label{prop: iter b2}
	Under Assumption \ref{asmp: iter b2}:
	\begin{enumerate}[(i)]
		\item When $B_2 \asymp (p/D)^{|S^*_{[0]}|+1}$, Assumption \ref{asmp: sure screening} holds.
		\item Let $D \geq p^*$. There exist $B_2 \asymp p/D$ in the first step, and $B_2\lesssim D^{|S^*_{[0]}|}(\log p)^{|S^*_{[1]}|-1} p$ in the second step, such that Assumption \ref{asmp: sure screening iterative} holds.
	\end{enumerate}
\end{proposition}

\subsection{More numerical experiment results}\label{subsec: rase numerical}

\subsubsection{Further illustration of ISIS used in numerical studies}\label{subsubsec: isis}
As described at the beginning of Section \ref{sec:numerical}, we report the results of ISIS with two screening steps and one selection step. To facilitate readers' understanding, we would like to explain the details via the following toy example. Suppose we have $n = 10$ observations with $p=10$ predictors, i.e. features 1-10. ISIS first screens 10 variables and selects the top $[n/\log n] \approx 4$ variables (screening step 1), which follows by fitting a LASSO model on these four variables and select those with non-zero coefficients (selection step 1). Suppose in screening step 1, the feature ranking is 4, 3, 1, 2, 6, 5, 7, 8, 10, 9. ISIS selects features 1-4 while in selection step 1, only features 1 and 4 have non-zero coefficients and are selected. Then ISIS screens on the remaining 8 features, and outputs the ranking of them (screening step 2). Suppose the ranking is 3, 7, 5, 9, 2, 6, 8, 10. In this case, the final output ranking from ISIS is 4, 1, 3, 7, 5, 9, 2, 6, 8, 10. The details can be summarized as follows:
\begin{itemize}
	\item \underline{Screening step 1:} 4, 3, 1, 2, 6, 5, 7, 8, 10, 9. Choose features 4, 3, 1, 2 to fit the LASSO.
	\item \underline{Selection step 1:} 4, 1. (LASSO selects features 4 and 1)
	\item \underline{Screening step 2:} 3, 7, 5, 9, 2, 6, 8, 10. (by screening on the remaining 8 features.)
	\item \underline{Final output ranking:} 4, 1, 3, 7, 5, 9, 2, 6, 8, 10.
\end{itemize}

\subsubsection{Additional simulations}\label{subsubsec: additional simulations}

\begin{example}[Example \ref{exp_sis} with realistic predictors, motivated by Model (3)(d) in \cite{nandy2020covariate}]\label{exp_realistic}
	In each replication, we randomly sample $p = 2000$ predictors from 18976 variables in rat eye expression data set, which is used in our real data example in Section \ref{subsubsec: rat}. Then we generate the response from the same model \eqref{eq: model in exp 1} in Example \ref{exp_sis}. Here the sample size $n = 120$, as in the rat eye expression data set.
\end{example}

The goal of this example is to demonstrate the effectiveness of RaSE for realistic predictors. The performances of different approaches are summarized in the left panel of Table \ref{table: exp realistic_mixed}, leading to similar findings as in Example \ref{exp_sis} with RaSE\textsubscript{1}-eBIC still achieving the best overall performance. Note that HOLP seems to have a much better performance than that in Example \ref{exp_sis}.

\begin{table}[!h]
\renewcommand{\arraystretch}{0.7}
\setlength{\tabcolsep}{7pt}
\begin{center}
\begin{threeparttable}
\begin{tabular}{l|rrrrr|rrrrr}
\Xhline{1pt}
\multirow{2}{*}{Method/MMS} & \multicolumn{5}{c}{Example 7} & \multicolumn{5}{c}{Example 8} \\ \cline{2-11}
&5\% &25\%  &50\%  &75\%  &95\% &5\% &25\%  &50\%  &75\%  &95\%  \\
\hline
SIS&  173&  660&  1184&  1634&  1907& 7&  41&  221&  599&  1847\\
ISIS&  15&  18&  24&  115&  896& 114&  850&  1392&  1806&  1969\\
SIRS&  564&  1174&  1537&  1798&  1944& 6&  1226&  1602&  1864&  1960\\
DC-SIS&  675&  1177&  1552&  1823&  1958& 61&  1153&  1620&  1831&  1976\\
HOLP&  4&  11&  44&  151&  657& 63&  270&  604&  1217&  1847\\
IPDC&  487&  936&  1334&  1642&  1940& 220&  588&  940&  1411&  1841\\
MDC-SIS&  110&  602&  1150&  1617&  1922& 7&  36&  148&  1218&  1915\\
CIS&  264&  703&  1222&  1644&  1936& 2000&  2000&  2000&  2000&  2000\\
RaSE-BIC&  4&  11&  34&  98&  1131& 8&  949&  1512&  1797&  1954\\
RaSE\textsubscript{1}-BIC&  4&  4&  4&  53&  101& 18&  1228&  1618&  1872&  1966\\
RaSE-eBIC&  4&  19&  59&  822&  1829& 7&  448&  1619&  1780&  1965\\
RaSE\textsubscript{1}-eBIC&  4&  4&  6&  10&  40& 1068&  1527&  1753&  1896&  1984\\
RaSE-$k$NN&  12&  109&  270&  1447&  1911& 5&  5&  10&  326&  1740\\
RaSE\textsubscript{1}-$k$NN&  4&  33&  658&  1412&  1834& 5&  5&  5&  842&  1858\\
RaSE-SVM&  7&  61&  204&  937&  1802& 5&  5&  8&  94&  1504\\
RaSE\textsubscript{1}-SVM&  4&  32&  72&  880&  1707& 5&  5&  5&  6&  675\\
\Xhline{1pt}
\end{tabular}
\end{threeparttable}
\caption{Quantiles of MMS in Examples \ref{exp_realistic} and \ref{exp_mixed}.}
\label{table: exp realistic_mixed}
\end{center}
\end{table}

\begin{example}[Example \ref{exp_knn} with mixed types of covariates]\label{exp_mixed}
	We generate $\tilde{\bx}' \sim N(\bm{0}, \Sigma)$ where $\Sigma = (\sigma_{ij})_{(4p/5) \times (4p/5)} = (0.5^{|i-j|})_{(4p/5) \times (4p/5)}$, $\tilde{\bx}' \in \mathbb{R}^{4p/5}$, and $p = 2000$. And we generate $\tilde{\bx}'' \sim \textup{Unif}(\{-2,-1,0,1,2\})$, where $\tilde{\bx}'' \in \mathbb{R}^{p/5}$. Then we construct $\tilde{\bm{x}}$ by letting $\tilde{\bx}_{S^c} = \tilde{\bx}'$ and $\tilde{\bx}_{S} = \tilde{\bx}''$, where $S = \{5, 10, \ldots, 1995, 2000\}$. Generate $z \sim \textup{Unif}(\{-3, 3\}) \ind \tilde{\bx}$ and $\bx = \tilde{\bx} + z\mathbf{1}_p$. The response is generated from the same model \eqref{eq: model in exp 3} in Example \ref{exp_mixed}. The signal set $S^* = \{1, 2, 3, 4, 5\}$. $n=200$.
\end{example}

This example is used to verify the effectiveness of RaSE in dealing with mixed types of covariates, which is very common in applications. Results are summarized in the right panel of Table \ref{table: exp realistic_mixed}, which are similar to those of Example \ref{exp_knn}. MDC-SIS and SIS still achieve a good performance at $5\%$ and $25\%$ quantiles, with RaSE\textsubscript{1}-SVM outperforming all other methods.

\begin{figure}[!h]
  \centering
  \subfloat[RaSE-BIC]{\label{fig: D_settings_pic}\includegraphics[width=.5\textwidth]{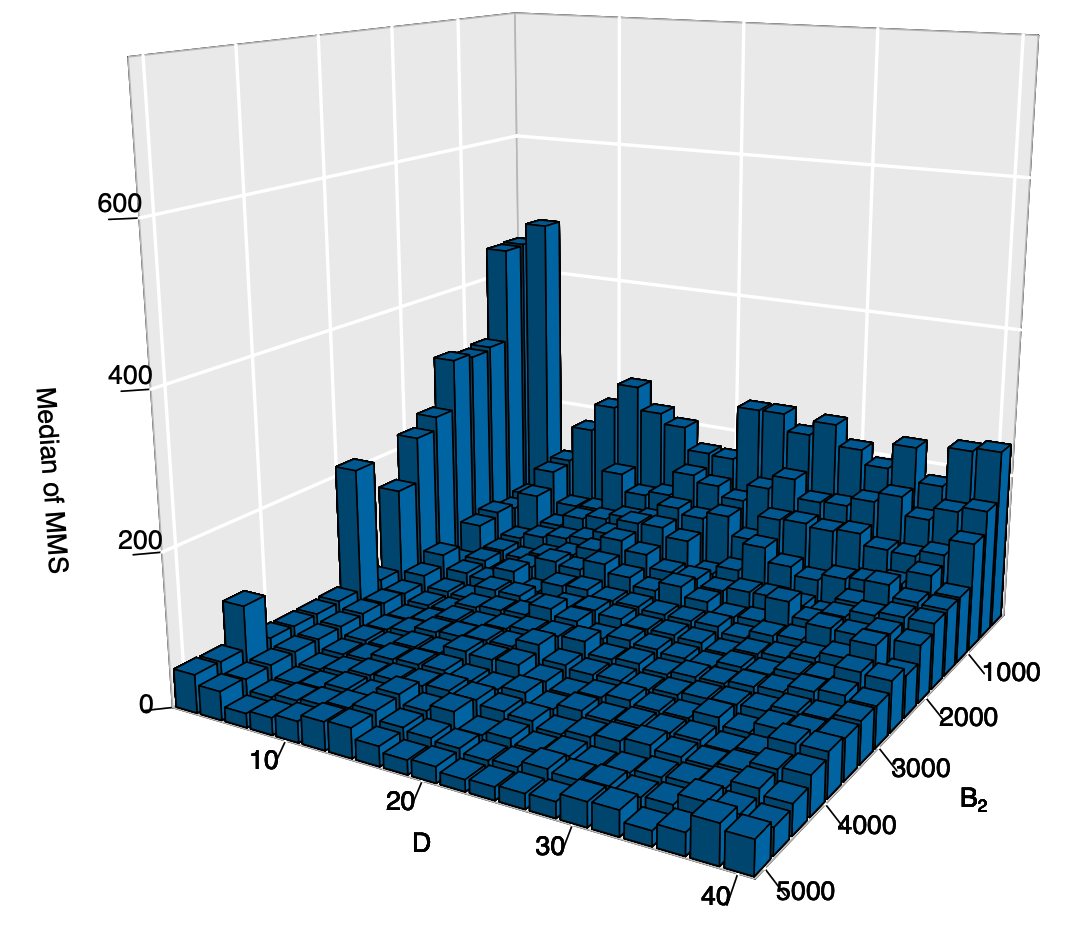}}
  \subfloat[RaSE\textsubscript{1}-BIC]{\label{fig: D_settings_1_pic}\includegraphics[width=.5\textwidth]{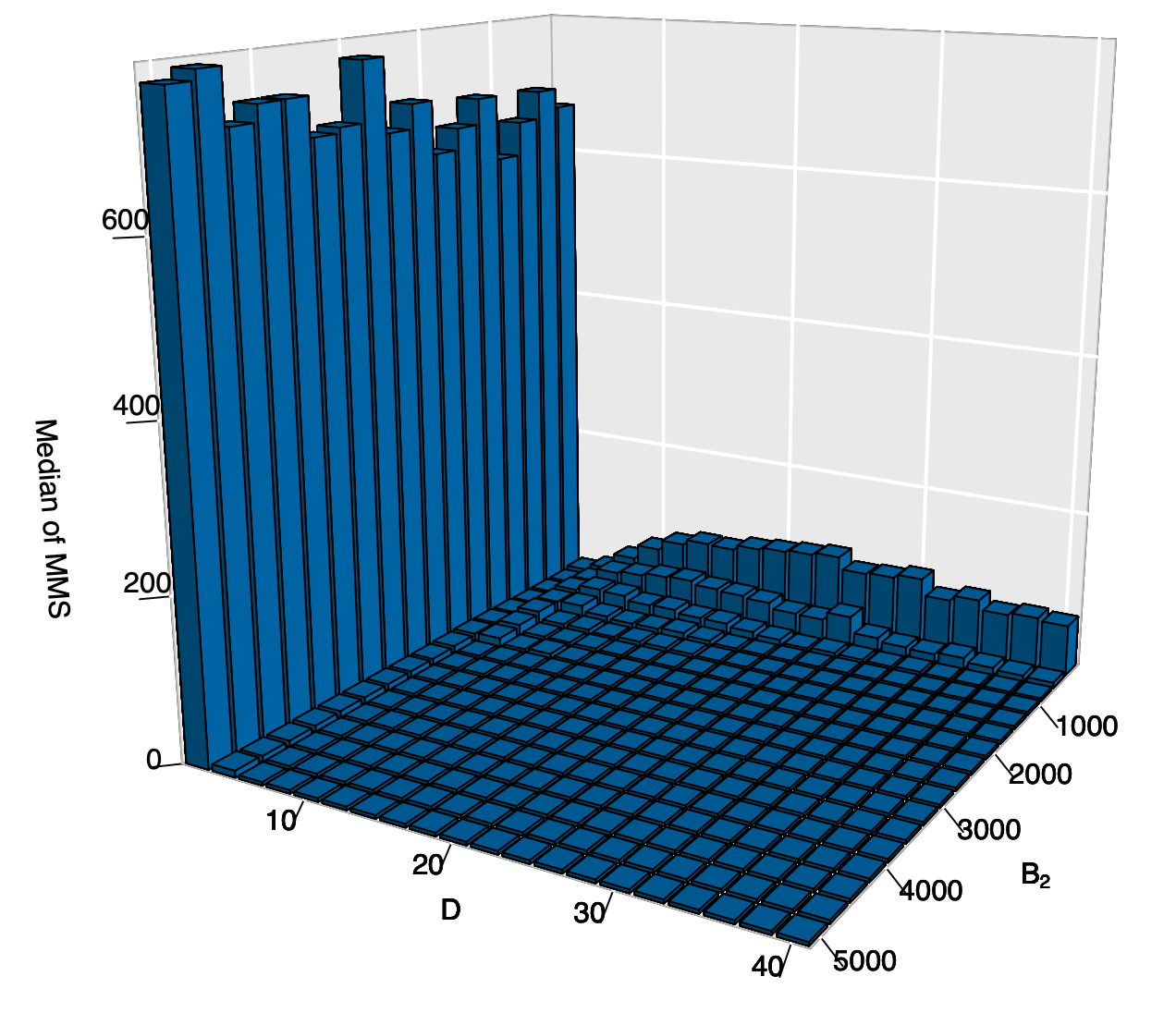}}
  \caption{Median MMS to capture $S^*$ ($|S^*| = 4$) as $(D, B_2)$ varies for RaSE-BIC (a) and RaSE\textsubscript{1}-BIC (b) in Example \ref{exp_sis}.}
  \label{fig: D_settings}
\end{figure}

\subsubsection{Additional figures}
See Figures \ref{fig: exp_samworth}, \ref{fig: alon_perc} and \ref{fig: rat_perc}. The full name of variables in Figure \ref{fig: alon_perc} can be found at \url{http://genomics-pubs.princeton.edu/oncology/affydata/names.html}

\begin{figure}[!h]
	\centering
	\includegraphics[width=0.8\textwidth]{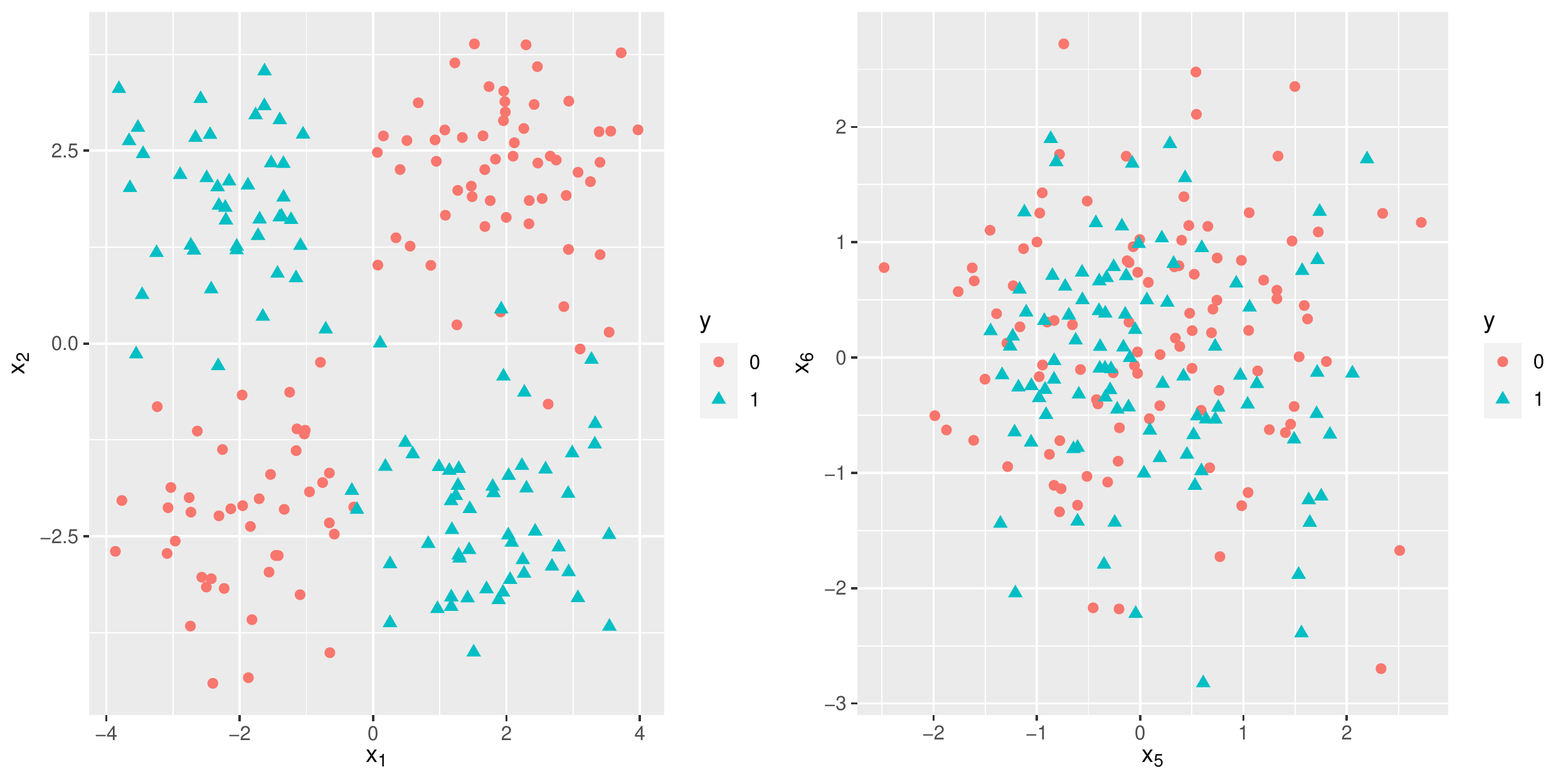}
	\caption{Scatterplots of $x_2$ vs. $x_1$ and $x_6$ vs. $x_{5}$ for Example \ref{exp_samworth} ($n = 200$).}
	\label{fig: exp_samworth}
\end{figure}

\begin{figure}[hp]
	\centering
	\includegraphics[width=\textwidth]{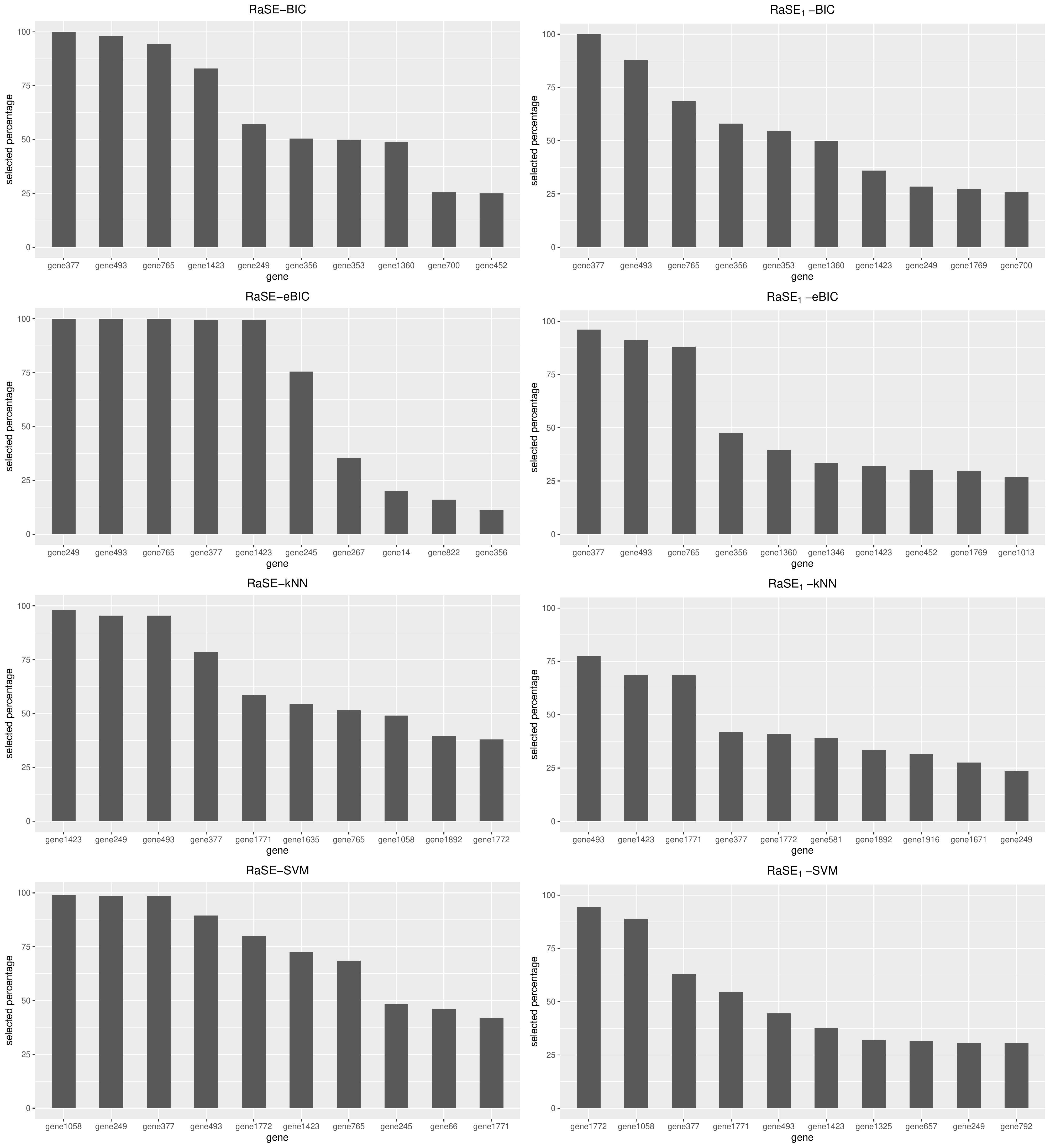}
	\caption{Features with the 10 highest selection rates (percentages in 200 replications) in the colon cancer data set.}
	\label{fig: alon_perc}
\end{figure}

\begin{figure}[hp]
	\centering
	\includegraphics[width=\textwidth]{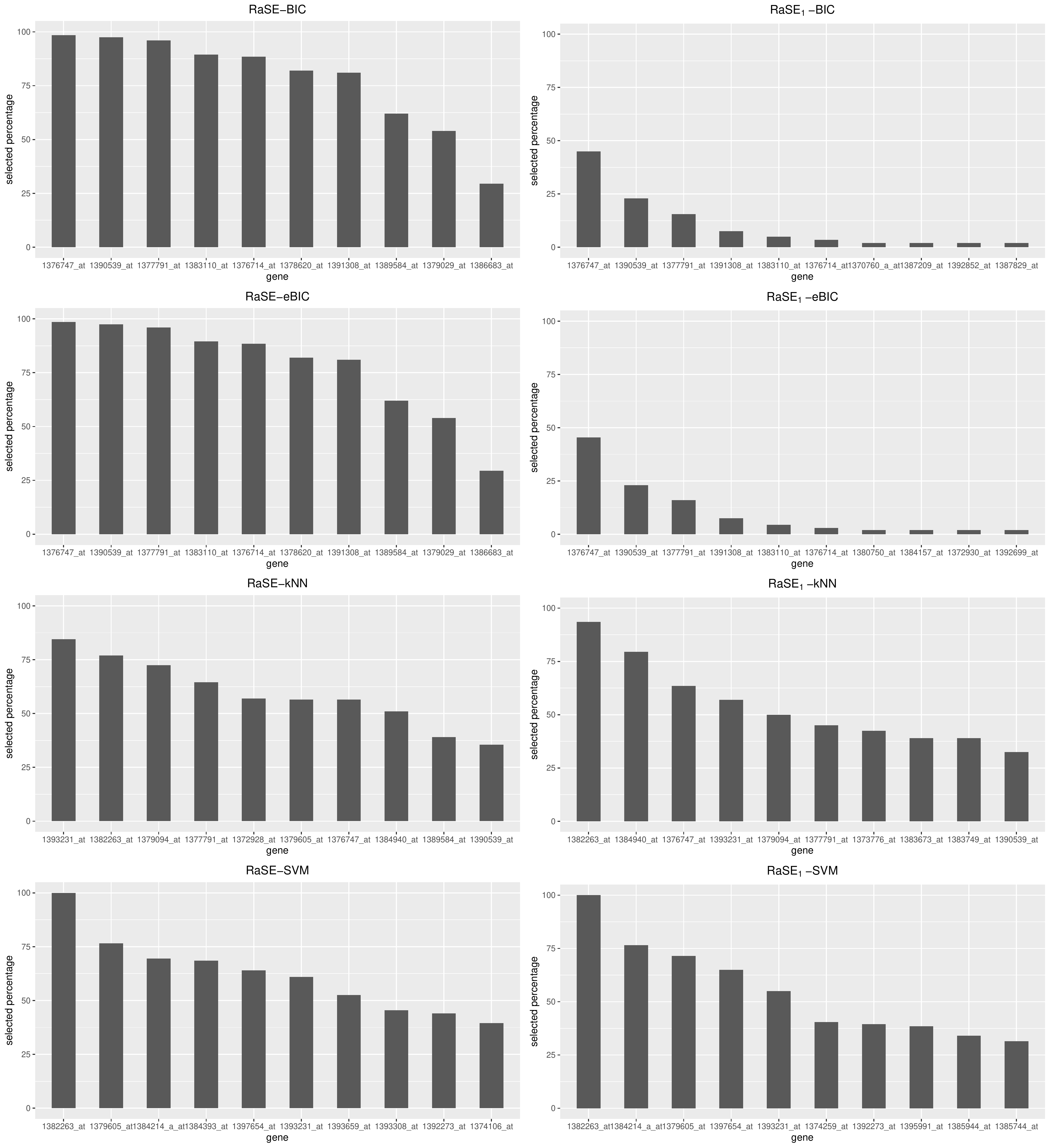}
	\caption{Features with the 10 highest selection rates (percentages in 200 replications) in the rat eye expression data set.}
	\label{fig: rat_perc}
\end{figure}

\section{Proofs}\label{sec: proofs}
In this section we use $C$ to represent a positive constant which is irrelevant to $n, D, B_1$ and $B_2$. It might take different values at different places.
\subsection{Proof of Lemma \ref{lem: ind}}
For arbitrary subsets $A_1^{(j)}, \ldots, A_{k}^{(j)}, A_{1}^{(-j)} \ldots, A_{B_2-k}^{(-j)}$, we have
	\begin{align}
		&\bp(\sj{11} = A_1^{(j)}, \ldots, \sj{1k} = A_k^{(j)}, \smj{11} = A_1^{(-j)}, \ldots, \smj{1(B_2-k)} = A_{B_2-k}^{(j)}|N_j = k) \\
		&= \frac{\bp(\sj{11} = A_1^{(j)}, \ldots, \sj{1k} = A_k^{(j)}, \smj{11} = A_1^{(-j)}, \ldots, \smj{1(B_2-k)} = A_{B_2-k}^{(j)})}{\bp(N_j = k)}\cdot \prod_{i=1}^k \mathds{1}(A_i^{(j)} \supseteq \bar{S}_j)\\
		&\quad \cdot \prod_{i=1}^{B_2-k} \mathds{1}(A_i^{(-j)} \not\supseteq \bar{S}_j)\\
		&= \frac{\binom{B_2}{k} \prod\limits_{i=1}^k \frac{1}{D}\frac{1}{\binom{p}{|A_{1i}^{(j)}|}}\prod\limits_{i=1}^{B_2-k} \frac{1}{D}\frac{1}{\binom{p}{|A_{1i}^{(-j)}|}}}{\binom{B_2}{k}p_j^k (1-p_j)^{B_2-k}}\cdot \prod_{i=1}^k \mathds{1}(A_i^{(j)} \supseteq \bar{S}_j)\cdot \prod_{i=1}^{B_2-k} \mathds{1}(A_i^{(-j)} \not\supseteq \bar{S}_j) \\
		&= \prod_{i=1}^k \frac{1}{D}\frac{1}{\binom{p}{|A_{1i}^{(j)}|}p_j}\prod_{i=1}^{B_2-k} \frac{1}{D}\frac{1}{\binom{p}{|A_{1i}^{(-j)}|}(1-p_j)}\cdot \prod_{i=1}^k \mathds{1}(A_i^{(j)} \supseteq \bar{S}_j)\cdot \prod_{i=1}^{B_2-k} \mathds{1}(A_i^{(-j)} \not\supseteq \bar{S}_j)\\
		&= \prod_{i=1}^k \frac{1}{D}\frac{1}{\binom{p}{|A_{1i}^{(j)}|}p_j}\cdot \mathds{1}(A_i^{(j)} \supseteq \bar{S}_j)\cdot \prod_{i=1}^{B_2-k} \frac{1}{D}\frac{1}{\binom{p}{|A_{1i}^{(-j)}|}(1-p_j)}\cdot \mathds{1}(A_i^{(-j)} \not\supseteq \bar{S}_j),
	\end{align}
	where the last equation is the joint density under conditions (\rom{1})-(\rom{3}). The second equality in (\rom{1}) and (\rom{2}) can be easily derived by basic algebra. This completes the proof.

\subsection{Proof of Theorem \ref{thm: sure screening}}

\begin{lemma}[Chernoff's bound]\label{lem: chernoff}
	For $N \sim \textup{Bin}(B_2, p_0)$, where $p_0 \in (0, 1)$, it holds
	\begin{align}
		\p(N > (1+\lambda)B_2p_0) &\leq \exp\left\{-\frac{\lambda^2}{2+\frac{2}{3}\lambda}B_2p_0\right\}, \lambda > 0,\\
		\p(N < (1-\lambda)B_2p_0) &\leq \exp\left\{-\frac{\lambda^2}{2+\frac{2}{3}\lambda}B_2p_0\right\}, \lambda \in (0, 1).
	\end{align}
\end{lemma}
\begin{proof}
	See \cite{shalev2014understanding}.
\end{proof}

\begin{lemma}[Hoeffding's inequality]\label{lem: hoeffding}
	Let $N_j \sim \textup{Bin}(B_2, p_0)$, where $p_0 \in (0, 1)$, then for arbitrary $x \in (0, p_0)$, there holds that
	\begin{equation}
		\max\{\p(N_j < B_2(p_0-x)), \p(N_j > B_2(p_0+x))\} \leq \exp\{-2B_2x^2\}.
	\end{equation}
\end{lemma}

Next we prove Theorem \ref{thm: sure screening} by applying the lemmas above. Denote the training data as $\mathcal{D} = \{(\bx_i, y_i)\}_{i=1}^n$ and event $\{\supp |\cri_n(S) - \cri(S)| \leq \epsilon_n\}$ as $\mathcal{A}$. Let $j$ be an arbitrary feature in $S^*$. Because of Lemma \ref{lem: ind}, it holds
\begin{align}
	&\p(j \in S_{1*}|N_j = k, \mathcal{A}) \\
	&= \p\left(\inf_{1 \leq b_2\leq k}\cri_n(\sj{1b_2}) < \inf_{1 \leq b_2\leq B_2-k}\cri_n(\smj{1b_2})|N_j = k, \mathcal{A}\right) \\
	&\geq \p\left(\inf_{1 \leq b_2\leq k}\cri(\sj{1b_2}) + 2\epsilon_n < \inf_{1 \leq b_2\leq B_2-k}\cri(\smj{11})|N_j = k, \mathcal{A}\right) \\
	&= \p\left(\inf_{1 \leq b_2\leq k}\cri(\sj{1b_2}) + 2\epsilon_n < \inf_{1 \leq b_2\leq B_2-k}\cri(\smj{11})|N_j = k\right) \\
	&= \be_{\{\sj{1b_2}\}_{b_2=1}^k} \left[\bp\left(\inf_{1 \leq b_2\leq k}\cri(\sj{1b_2}) + 2\epsilon_n < \inf_{1 \leq b_2\leq B_2-k}\cri(\smj{11})|N_j = k, \{\sj{1b_2}\}_{b_2=1}^k\right)\right]\\
	&= \be_{\{\sj{1b_2}\}_{b_2=1}^k} \left[\bp\left(\inf_{1 \leq b_2\leq k}\cri(\sj{1b_2}) + 2\epsilon_n < \cri(\smj{11})|N_j = k, \{\sj{1b_2}\}_{b_2=1}^k\right)^{B_2-k}\right]\\
	&\geq \bp\left(\inf_{1 \leq b_2\leq k}\cri(\sj{1b_2}) + 2\epsilon_n < \cri(\smj{})|N_j = k\right)^{B_2-k} \label{eq: jensen} \\
	&= \left[1-\bp\left(\inf_{1 \leq b_2\leq k}\cri(\sj{1b_2}) + 2\epsilon_n \geq \cri(\smj{})|N_j = k\right)\right]^{B_2-k} \\
	&= \left\{1-\be_{\smj{}}\left[\bp\left(\inf_{1 \leq b_2\leq k}\cri(\sj{1b_2}) + 2\epsilon_n \geq \cri(\smj{11})|N_j = k, \smj{}\right)\right]\right\}^{B_2-k} \\
	&= \left\{1-\be_{\smj{}}\left[\bp_{\sj{}}\left(\cri(\sj{}) + 2\epsilon_n \geq \cri(\smj{11})|N_j = k, \smj{}\right)^{k}\right]\right\}^{B_2-k} \\
	&= (1-\be_{\smj{}}[\delta_{j}(\smj{})^k])^{B_2-k}.\label{eq: ind multi}
\end{align}
where \eqref{eq: jensen} holds due to Jensen inequality. It follows that
\begin{equation}\label{eq: j 1star}
	\p(j \in S_{1*}|N_j = k) \geq \p(j \in S_{1*}|N_j = k, \mathcal{A})\tp(\mathcal{A})
	\geq (1-c_{1n})(1-\be_{\smj{}}[\delta_{j}(\smj{})^k])^{B_2-k}.
\end{equation}
Then we have
\begin{align}
	\p(j \in S_{1*}) &= \sum_{k=1}^{B_2}\p(j \in S_{1*}|N_j = k)\bp(N_j = k)\\
	&\geq (1-c_{1n})\sum_{k=1}^{B_2}(1-\be_{\smj{}}[\delta_{j}(\smj{})^k])^{B_2-k}\cdot\bp(N_j = k) \\
	&\geq (1-c_{1n})\sum_{k\geq \frac{1}{2}B_2p_0}(1-\be_{\smj{}}[\delta_{j}(\smj{})^k])^{B_2-k}\cdot\bp(N_j = k) \\
	&\geq (1-c_{1n})(1-\be_{\smj{}}[\delta_{j}(\smj{})^{\frac{1}{2}B_2p_0}])^{B_2}\cdot \bp\left(N_j \geq \frac{1}{2}B_2p_j\right).\label{eq: js}
\end{align}
By letting $\lambda = \frac{1}{2}$ in Lemma \ref{lem: chernoff}, we have
\begin{equation}\label{eq: chf}
	\bp\left(N_j \geq \frac{1}{2}B_2p_j\right) \geq 1-\exp\left\{-\frac{3}{28}B_2p_j\right\}.
\end{equation}
Combining \eqref{eq: js} and \eqref{eq: chf}, we obtain
\begin{equation}\label{eq: cover signal}
	\eta_j = \p(j \in S_{1*}) \geq c_{2n}.
\end{equation}
Again, by Lemma \ref{lem: hoeffding} and union bounds, it holds
\begin{align}
	\p(S^* \subseteq \hat{S}_{\alpha}) &\geq \p\left(\bigcap_{j \in S^*}\left\{\heta_j \geq \frac{c_{2n}}{\alpha}\right\} \right) \\
	&\geq 1-\sum_{j \in S^*}\p\left(\heta_j -\eta_j < -\left(1-\frac{1}{\alpha}\right)\eta_j \right)\\
	&\geq 1-\sum_{j \in S^*}\exp\left\{-2B_1\eta_j^2\left(1-\frac{1}{\alpha}\right)^2\right\} \\
	&\geq 1-p^*\exp\left\{-2B_1 c_{2n}^2\left(1-\frac{1}{\alpha}\right)^2\right\}, \label{eq: hoef union bdds}
\end{align}
which leads to conclusion (\rom{1}). 

For (\rom{2}), since $\limsup\limits_{n,D,B_2 \rightarrow \infty} \left\{B_2\sup_{j \in S^*}\be_{\smj{}}[\delta_{j}(\smj{})^{\frac{1}{2}B_2p_0}]\right\} < \infty$, $c_{1n} \rightarrow 0$ and $B_2\inf\limits_{j \in S^*}p_j \gtrsim 1$, it's easy to show that
\begin{equation}
	\liminf\limits_{n,D,B_2 \rightarrow \infty} c_{2n} > 0,
\end{equation}
yielding $|\hat{S}_{\alpha}| \lesssim D$, which completes our proof.

\subsection{Proof of Proposition \ref{prop: sure screening joint contribution}}
\begin{lemma}\label{lem: pj}
	$p_j = \bp(S_{11} \supseteq \bar{S}_j) = D^{-1}\sum_{d = |\bar{S}_j|}^{D}\frac{\binom{p-|\bar{S}_j|}{d-|\bar{S}_j|}}{\binom{p}{d}} \asymp \left(\frac{p}{D}\right)^{|\bar{S}_j|}$, where $S_{11}$ follows the distribution in \eqref{eq: sj dist} w.r.t. $\bar{S}_j$.
\end{lemma}

\begin{proof}[Proof of Lemma \ref{lem: pj}]
	First, notice that
	\begin{equation}
		p_j = \frac{1}{D}\sum_{d=|\bar{S}_j|}^D \frac{(p-|\bar{S}_j|)\cdots (p-d+1)}{p\cdots (p-d+1)}\cdot d \cdots (d-|\bar{S}_j|+1)
		\leq \frac{1}{D}\sum_{d=|\bar{S}_j|}^D \frac{D^{|\bar{S}_j|}}{(p-D+1)^{|\bar{S}_j|}} 
		\lesssim \left(\frac{p}{D}\right)^{|\bar{S}_j|}.
	\end{equation}
	On the other hand,
	\begin{align}
		p_j &\geq \frac{1}{Dp^{|\bar{S}_j|}}\cdot \left(1-\frac{|\bar{S}_j|}{p-D}\right)^D\sum_{d=|\bar{S}_j|}^D d \cdots (d-|\bar{S}_j|+1) \\
		&\geq \frac{1}{Dp^{|\bar{S}_j|}}\cdot \left(1-\frac{|\bar{S}_j|}{p-D}\right)^D\sum_{r=1}^{D-|\bar{S}_j|+1} r\cdots(r+|\bar{S}_j|-1) \\
		&\gtrsim \left(\frac{p}{D}\right)^{|\bar{S}_j|},
	\end{align}
	where the last inequality is because $D \ll p$, $|\bar{S}_j| \leq p^*$ and $p^*$ is fixed.
\end{proof}

\begin{lemma}\label{lem: prop sure screening}
	Under Assumption \ref{asmp: sure screening joint contribution}, we have
	\begin{equation}
		\bp_{\smj{}}(|\smj{} \cap (S^* \cup S_j^0)| \geq |\bar{S}_j|) \lesssim \left(\frac{D}{p}\right)^{|\bar{S}_j|},
	\end{equation}
	where $\smj{}$ follows the distribution in \eqref{eq: smj dist}.
\end{lemma}

\begin{proof}[Proof of Lemma \ref{lem: prop sure screening}]
	By Lemma \ref{lem: pj}, when $n$ is sufficiently large, $\inf\limits_{j \in S^*}p_j < \frac{1}{2}$. Then for any integer $t \in [0, p^*+p_j^0]$, we have
	\begin{align}
		&\bp_{\smj{}}(|\smj{} \cap (S^* \cup S_j^0)| = t) \\
		&= \sum_{d = t}^D \frac{1}{D}\frac{1}{\binom{p}{d}(1-p_j)}\cdot \binom{p-p^*-p_j}{d-t}\binom{p^*+p_j^0}{t} \\
		&\leq \frac{2}{D}\sum_{d = t}^D \frac{(p-p^*-p_j^0)\cdots (p-p^*-p_j^0 -d +t + 1)}{p\cdots (p-d+1)} \cdot \binom{p^*+p_j^0}{t} d\cdots (d-t+1)\\
		&\leq \frac{2}{D}\sum_{d = t}^D \left(1-\frac{p^*+p^0_j}{p}\right)^d\cdot \frac{1}{(p-D+1)^t}\cdot D^t \\
		&\leq \frac{2\binom{p^*+p_j^0}{t}D^t}{(p-D+1)^t} \\
		&\lesssim  \left(\frac{D}{p}\right)^t.
	\end{align}
	Therefore,
	\begin{equation}
		\bp_{\smj{}}(|\smj{} \cap (S^* \cup S_j^0)| \geq |\bar{S}_j|) \lesssim \bp_{\smj{}}(|\smj{} \cap (S^* \cup S_j^0)| = |\bar{S}_j|) \lesssim \left(\frac{D}{p}\right)^{|\bar{S}_j|}.
	\end{equation}
\end{proof}

Now let's prove Proposition \ref{prop: sure screening joint contribution}. Due to Assumption \ref{asmp: sure screening joint contribution}, for any $j$ there exists a set $\bar{S}_j \ni j$ with cardinality $\mathpzc{d}$ to make $j$ detectable. By Lemmas \ref{lem: pj} and \ref{lem: prop sure screening}, it's easy to see that
\begin{equation}
	B_2\sup_{j \in S^*}\be_{\smj{}}[\delta_{j}(\smj{})^{\frac{1}{2}B_2p_j}] \leq B_2\sup_{j \in S^*} \bp_{\smj{}}(|\smj{} \cap (S^* \cup S_j^0)| \geq \mathpzc{d}) \lesssim B_2\left(\frac{D}{p}\right)^{\mathpzc{d}} < \infty,
\end{equation}
\begin{equation}
	B_2\inf_{j \in S^*}p_j \gtrsim B_2\left(\frac{p}{D}\right)^{\mathpzc{d}} \gtrsim 1,
\end{equation}
implying that conditions in Assumption \ref{asmp: sure screening joint contribution} hold.

\subsection{Proof of Theorem \ref{thm: sure screening iterative}}
Akin to the argument in the proof of Theorem \ref{thm: sure screening}, we can prove that
\begin{equation}
	\p\left(\bigcap_{j \in S^*_{\mcf}}\left\{\heta_j^{\mcf} \geq c_2^*\right\} \right)\geq 1-p^*\exp\left\{-2B_1 (c_{2n}^{\mcf}-c_2^*)^2\right\}.
\end{equation}
For convenience, denote event $\bigcap_{j \in S^*_{\mcf}}\left\{\heta_j^{\mcf} \geq c_2^*\right\}$ as $\mathcal{B}$. Write $\#\{b_2: j \in S_{1b_2}^{\mcs}\}$ as $N_j$, then similar to \eqref{eq: ind multi} and \eqref{eq: j 1star}, since
\begin{equation}
	\inf_{i \in S^*_{\mcf}} \tilde{\eta}_i \geq \frac{c_2^*}{(D+C_0)},
\end{equation}
the subspace distribution family satisfies the condition in Assumption \ref{asmp: sure screening iterative}.(\rom{2}), it follows
\begin{equation}
	\p(j \in S_{1*}^{\mcs}|N_j = k, \mathcal{B}) \geq (1-c_{1n})(1-(\delta_n^{\mcs})^k)^{B_2-k}.
\end{equation}
Similar to \eqref{eq: js}, it follows that
\begin{align}
	&\p(j \in S_{1*}^{\mcs}|\hat{\bm{\eta}}^{\mcf} \in \mathcal{B}) \\
	&= \sum_{k=1}^{B_2}\p(j \in S_{1*}^{\mcs}|N_j = k, \mathcal{B})\bp(N_j = k|\hat{\bm{\eta}}^{\mcf} \in \mathcal{B})\\
	&\geq (1-c_{1n})\sum_{k=1}^{B_2}(1-\be_{\smj{}}[\delta_{j}(\smj{})^k])^{B_2-k}\cdot\bp(N_j = k|\hat{\bm{\eta}}^{\mcf} \in \mathcal{B}) \\
	&\geq (1-c_{1n})\sum_{k\geq \frac{1}{2}B_2p_j^{\mcs}}(1-\be_{\smj{}}[\delta_{j}(\smj{})^k])^{B_2-k}\cdot\bp(N_j = k|\hat{\bm{\eta}}^{\mcf} \in \mathcal{B}) \\
	&\geq (1-c_{1n})(1-\be_{\smj{}}[\delta_{j}(\smj{})^{\frac{1}{2}B_2p_j^{\mcs}}])^{B_2}\cdot \bp\left(N_j \geq \frac{1}{2}B_2p_j^{\mcs}\bigg|\hat{\bm{\eta}}^{\mcf} \in \mathcal{B}\right),\label{eq: jsp 2}
\end{align}
where $p_j^{\mcs} = \bp_{S_{11}^{\mcs} \sim \mathcal{R}(\mathcal{U}_0, p, \tilde{\bm{\eta}}^{\mcf})}(j \in S_{11}^{\mcs}|\hat{\bm{\eta}}^{\mcf} \in \mathcal{B})$. Notice that $N_j \sim \textup{Bin}(B_2, p_j^{\mcs})$ when $\hat{\bm{\eta}}^{\mcf}$ is given. 
Therefore, by Lemma \ref{lem: chernoff}, analogous to the proof of Theorem \ref{thm: sure screening}, we can obtain
\begin{equation}\label{eq: ineq 1}
	\bp\left(N_j \geq \frac{1}{2}B_2p_j^{\mcs}\bigg|\hat{\bm{\eta}}^{\mcf} \in \mathcal{B}\right) 
	\geq 1-\exp\left\{-\frac{3}{28}B_2p_j^{\mcs}\right\}.
\end{equation}
Therefore, together with \eqref{eq: ineq 1}, equation \eqref{eq: jsp 2} leads to
\begin{align}
	&\p(j \in S_{1*}^{\mcs}|\hat{\bm{\eta}}^{\mcf} \in \mathcal{B}) \\
	&\geq (1-c_{1n})(1-\be_{\smj{}}[\delta_{j}(\smj{})^{\frac{1}{2}B_2p_j^{\mcs}}])^{B_2}\left(1-\exp\left\{-\frac{3}{28}B_2p_j^{\mcs}\right\}\right). \label{eq: jsp 3}
\end{align}
Therefore, given $\hat{\bm{\eta}}^{\mcf} \in \mathcal{B}$, it follows that
\begin{equation}
	\eta_j^{\mcs} =  \p(j \in S_{1*}^{\mcs}|\hat{\bm{\eta}}^{\mcf} \in \mathcal{B}) \geq c_{2n}^{\mcs}.
\end{equation}
And for any $j$, it's straightforward to see that 
\begin{align}
	\bp_{S_{11}^{\mcs} \sim \mathcal{R}(\mathcal{U}_0, p, \tilde{\bm{\eta}}^{\mcf})}(j \in S_{11}^{\mcs}|\hat{\bm{\eta}}^{\mcf} \in \mathcal{B}, |S_{11}^{\mcs}| = d) &\geq 1-\left[1-\frac{\frac{C_0}{p}}{D+C_0}\right]\cdot \left[1-\frac{\frac{C_0}{p}}{D+C_0-\frac{C_0}{p}}\right] \cdots \\
	&\quad \left[1-\frac{\frac{C_0}{p}}{D+C_0-\frac{(d-1)C_0}{p}}\right]\\
	&\geq \frac{C_0d}{(D+C_0)p},
\end{align}
yielding that
\begin{equation}
	p_j^{\mcs} = \bp_{S_{11}^{\mcs} \sim \mathcal{R}(\mathcal{U}_0, p, \tilde{\bm{\eta}}^{\mcf})}(j \in S_{11}^{\mcs}|\hat{\bm{\eta}}^{\mcf} \in \mathcal{B}) \geq \frac{1}{D}\sum_{d=1}^D \frac{C_0d}{(D+C_0)p} \geq \frac{(D+1)C_0}{2(D+C_0)p} = p^{\mcs}.\label{eq: s11 j}
\end{equation}
Analogous to \eqref{eq: hoef union bdds}, applying Hoeffding's inequality and union bounds, we obtain
\begin{align}
	\p(S^* \subseteq \hat{S}_{\alpha}^{\mcs}) &\geq \p\left(\bigcap_{j \in S^*}\left\{\heta_j^{\mcs} \geq \frac{c_2^{\mcs}}{\alpha}\right\} \bigg| \hat{\bm{\eta}}^{\mcf} \in \mathcal{B}\right) \p(\hat{\bm{\eta}}^{\mcf} \in \mathcal{B})\\
	&\geq 1-\sum_{j \in S^*}\p\left(\heta_j^{\mcs} -\eta_j^{\mcs} < -\left(1-\frac{1}{\alpha}\right)c_2^{\mcs} \bigg|\hat{\bm{\eta}}^{\mcf} \in \mathcal{B}\right) - \p(\hat{\bm{\eta}}^{\mcf} \notin \mathcal{B})\\
	&\geq 1-p^*\exp\left\{-2B_1 (c_{2n}^{\mcf}-c_2^*)^2\right\} - p^*\exp\left\{-2B_1 (c_{2n}^{\mcs})^2\left(1-\frac{1}{\alpha}\right)^2\right\},
\end{align}
which completes the proof of conclusion (\rom{1}). The proof of conclusion (\rom{2}) is similar to that of Theorem \ref{asmp: sure screening}, and we omit it here.

\subsection{Proof of Theorem \ref{thm: rank consistency}}
Similar to \eqref{eq: ind multi}, for any $j \notin S^*$, we have
\begin{align}
	&\p(j \in S_{1*}|N_j = k, \mathcal{A})\\
	&= \p\left(\inf_{1 \leq b_2\leq k}\cri_n(\sj{1b_2}) < \inf_{1 \leq b_2\leq B_2-k}\cri_n(\smj{1b_2})|N_j = k, \mathcal{A}\right) \\
	&\leq \p\left(\inf_{1 \leq b_2\leq k}\cri(\sj{1b_2}) - 2\epsilon_n < \inf_{1 \leq b_2\leq B_2-k}\cri(\smj{11})|N_j = k, \mathcal{A}\right) \\
	&= \p\left(\inf_{1 \leq b_2\leq k}\cri(\sj{1b_2}) - 2\epsilon_n < \inf_{1 \leq b_2\leq B_2-k}\cri(\smj{11})|N_j = k\right) \\
	&= 1-\be_{\{\smj{1b_2}\}_{b_2=1}^{B_2-k}} \left[\bp\left(\inf_{1 \leq b_2\leq k}\cri(\sj{1b_2}) - 2\epsilon_n \geq \inf_{1 \leq b_2\leq B_2-k}\cri(\smj{11})|N_j = k, \{\smj{1b_2}\}_{b_2=1}^{B_2-k}\right)\right]\\
	&= 1-\be_{\{\smj{1b_2}\}_{b_2=1}^{k}} \left[\bp\left(\cri(\sj{}) - 2\epsilon_n \geq \inf_{1 \leq b_2\leq k}\cri(\smj{1b_2})|N_j = k, \{\smj{1b_2}\}_{b_2=1}^{B_2-k}\right)^{k}\right]\\
	&\leq 1-\bp\left(\cri(\sj{}) - 2\epsilon_n \geq \inf_{1 \leq b_2\leq B_2-k}\cri(\smj{1b_2})|N_j = k, \{\smj{1b_2}\}_{b_2=1}^{B_2-k}\right)^k \\
	&= 1-\left(1-\be_{\sj{}}\left[\bp_{\smj{}}(\cri(\sj{}) - 2\epsilon_n \geq \cri(\smj{})|\sj{})^{B_2-k}\right]\right)^k \\
	&= 1-\left(1-\be_{\sj{}}\left[\tilde{\delta}_j(\sj{})^{B_2-k}\right]\right)^k.
\end{align}
Then we have
\begin{align}
	\p(j \in S_{1*}|N_j = k) &\leq \p(j \in S_{1*}|N_j = k, \mathcal{A})\tp(\mathcal{A}) + \tp(\mathcal{A}^c)\\
	&\leq (1-c_{1n})\left\{1-\left(1-\be_{\sj{}}\left[\tilde{\delta}_j(\sj{})^{B_2-k}\right]\right)^k\right\} + c_{1n},
\end{align}
which combined with Lemma \ref{lem: chernoff} yields
\begin{align}
	\gamma' &\coloneqq \inf_{j \in S^*}\p(j \in S_{1*}) - \sup_{j \notin S^*}\p(j \in S_{1*})\\
	&\geq \inf_{j \in S^*}\sum_{k=1}^{B_2}\p(j \in S_{1*}|N_j = k)\bp(N_j = k)- \sup_{j \notin S^*}\sum_{k=1}^{B_2}\p(j \in S_{1*}|N_j = k)\bp(N_j = k)\\
	&\geq (1-c_{1n})\inf_{j \in S^*}\sum_{k=1}^{B_2}\left(1-\be_{\smj{}}[\delta_{j}(\smj{})^{k}]\right)^{B_2-k}\bp(N_j = k)\\
	&\quad +(1-c_{1n})\inf_{j \notin S^*}\sum_{k=1}^{B_2}\left[\left(1-\be_{\sj{}}\left[\tilde{\delta}_j(\sj{})^{B_2-k}\right]\right)^k-1\right]\cdot\bp(N_j = k) - c_{1n} \\
	&\geq (1-c_{1n})\inf_{j \in S^*}\sum_{k \geq \frac{1}{2}B_2p_j}\left(1-\be_{\smj{}}[\delta_{j}(\smj{})^{k}]\right)^{B_2-k}\bp(N_j = k)\\
	&\quad +(1-c_{1n})\inf_{j \notin S^*}\sum_{k\leq \frac{3}{2}B_2p^{[0]}}\left[\left(1-\be_{\sj{}}\left[\tilde{\delta}_j(\sj{})^{B_2-k}\right]\right)^k-1\right]\cdot\bp(N_j = k) - c_{1n} \\
	&\geq  \left[\left(1-\sup_{j \in S^*}\be_{\smj{}}[\delta_{j}(\smj{})^{\frac{1}{2}B_2p_j}]\right)^{B_2}+\left(1-\sup_{j \notin S^*}\be_{\sj{}}\left[\tilde{\delta}_j(\sj{})^{B_2-\frac{3}{2}B_2p^{[0]}}\right]\right)^{\frac{3}{2}B_2p^{[0]}}-1\right] \\
	&\quad \cdot(1-c_{1n})\left(1-2\exp\left\{-\frac{3}{28}B_2\inf_{j \in S^*}p_j\right\}\right) - c_{1n} \\
	&= \gamma.
\end{align}
Denote $\gamma'' = \frac{\inf_{j \in S^*}\eta_j - \sup_{j \notin S^*}\eta_j}{\inf_{j \in S^*}\eta_j}$. Then we have 
\begin{equation}
	\sup_{j \notin S^*}\eta_j = \sup_{j \notin S^*}\p(j \in S_{1*}) = (1-\gamma'') \inf_{j \in S^*}\p(j \in S_{1*}) = (1-\gamma'')\inf_{j \in S^*}\eta_j.
\end{equation}
By Hoeffding's inequality and union bounds, it follows
\begin{align}
	&\p\left(\inf_{j \in S^*} \heta_j > \sup_{j \notin S^*} \heta_j\right)\\
    &\geq \p\left(\sup_{j \notin S^*}\heta_j < \left(1-\frac{1}{2}\gamma''\right)\inf_{j \in S^*}\eta_j < \inf_{j \in S^*}\heta_j\right)\\
	&\geq 1-\p\left(\bigcup_{j\notin S^*}\left\{\heta_j - \eta_j\geq \frac{\gamma''}{2}\inf_{j \in S^*}\eta_j \right\}\right)-\p\left(\bigcup_{j\in S^*}\left\{\heta_j - \eta_j \leq  -\frac{\gamma''}{2}\inf_{j \in S^*}\eta_j \right\}\right) \\
	&\geq 1-\sum_{j\notin S^*}\p\left(\heta_j - \eta_j\geq \frac{\gamma'}{2}\right) - \sum_{j\in S^*}\p\left(\heta_j - \eta_j\leq -\frac{\gamma'}{2}\right)\\
	&\geq 1-\sum_{j\notin S^*}\exp\left\{-\frac{1}{2}B_1 \gamma'^2\right\}- \sum_{j\in S^*}\exp\left\{-\frac{1}{2}B_1 \gamma'^2\right\}\\
	&\geq 1-p\exp\left\{-\frac{1}{2}B_1 \gamma'^2\right\} \\
	&\geq 1-p\exp\left\{-\frac{1}{2}B_1 \gamma^2(n, D, B_2)\right\},
\end{align}
which converges to 1 due to Assumption \ref{asmp: rank consistency} as $n, B_1, B_2 \rightarrow \infty$. This completes the proof.

\subsection{Proof of Proposition \ref{prop: rank consistency}}
By the definition of $\gamma(n, D, B_2)$, it suffices to show that there exist constants $C_2>C_1>0$, such that
\begin{equation}
	\left(1-\sup_{j \in S^*}\be_{\smj{}}[\delta_{j}(\smj{})^{\frac{1}{2}B_2p_j}]\right)^{B_2}+\left(1-\sup_{j \notin S^*}\be_{\sj{}}\left[\tilde{\delta}_j(\sj{})^{B_2-\frac{3}{2}B_2p^{[0]}}\right]\right)^{\frac{3}{2}B_2p^{[0]}}>1,
\end{equation}
as well as $B_2\inf\limits_{j \in S^*}p_j \gtrsim 1$. Without loss of generality, assume feature $1 \in S^*$.

In fact, when $B_2 \in (\frac{C_1p}{D}, \frac{C_2p}{D})$, by the proof of Proposition \ref{prop: sure screening joint contribution},
\begin{equation}
	\left(1-\sup_{j \in S^*}\be_{\smj{}}[\delta_{j}(\smj{})^{\frac{1}{2}B_2p_j}]\right)^{B_2} \gtrsim e^{-C_2}.
\end{equation}
On the other hand, 
\begin{align}
	\sup_{j \notin S^*}\be_{\sj{}}\left[\tilde{\delta}_j(\sj{})^{B_2-\frac{3}{2}B_2p^{[0]}}\right] &\leq \sup_{j \notin S^*}\left[\bp_{\sj{}}(\sj{} \cap (\{1\}\cup S_1^0) \neq \emptyset) + \bp_{\smj{}}(1 \notin \smj{})^{B_2-\frac{3}{2}B_2p^{[0]}}\right] \\
	&\lesssim  \frac{D}{p} + \left(1-\frac{D}{p}\right)^{B_2-\frac{3}{2}B_2p^{[0]}}\\ 
	&\lesssim  e^{-C_1}. 
\end{align}
Therefore, $\exists C > 0$, such that
\begin{equation}
	\left(1-\sup_{j \notin S^*}\be_{\sj{}}\left[\tilde{\delta}_j(\sj{})^{B_2-\frac{3}{2}B_2p^{[0]}}\right]\right)^{\frac{3}{2}B_2p^{[0]}} \geq (1-Ce^{-C_1})^{C_2}.
\end{equation}
Let $C_1 = C_2/2$, the inequality above leads to
\begin{align}
	&\left(1-\sup_{j \in S^*}\be_{\smj{}}[\delta_{j}(\smj{})^{\frac{1}{2}B_2p_j}]\right)^{B_2}+\left(1-\sup_{j \notin S^*}\be_{\sj{}}\left[\tilde{\delta}_j(\sj{})^{B_2-\frac{3}{2}B_2p^{[0]}}\right]\right)^{\frac{3}{2}B_2p^{[0]}} \\
	&\geq  e^{-C_2} + (1-Ce^{-\frac{C_2}{2}})^{C_2} \\
	&> 1,
\end{align}
when $C_2$ is small than some constant, which finishes our proof.

\subsection{Proof of Proposition \ref{prop: iter b2}}
Note that $\mathpzc{d} = |S^*_{[0]}|+1$ and Assumption \ref{asmp: sure screening joint contribution} holds. Therefore conclusion (\rom{1}) holds due to Proposition \ref{prop: sure screening joint contribution}. Let's prove (\rom{2}). The $B_2$ requirement in first step is derived from Proposition \ref{prop: sure screening joint contribution} since $\sup_{i\in S_{[0]}^*}\mathpzc{d}_i = 1$. We only need to calculate the restriction on $B_2$ in the second step. By the proof of Theorem \ref{thm: sure screening iterative}, $\inf_{i \in S^*_{\mcf}} \hat{\eta}_i \geq c_2^*$. Recalling Algorithm \ref{algo: iterative rase screening}, we denote $\tilde{S}_{[0]}^* = \{i: \hat{\eta}^{[0]} > C_0/\log p \} \backslash S^*_{[0]}$. Then similar to \eqref{eq: s11 j}, we know that
\begin{equation}\label{eq: eta lower bdd}
	\inf_{i \in S^*_{\mcf}} \tilde{\eta}^{[0]}_i \gtrsim D^{-1}, \quad \inf_{i \in \tilde{S}_{[0]}^*} \tilde{\eta}^{[0]}_i \gtrsim (\log p)^{-1}, \quad \inf_{i \in S^*_{[1]} \backslash \tilde{S}_{[0]}^*} \tilde{\eta}^{[0]}_i \gtrsim p^{-1}.
\end{equation}
\begin{itemize}
	\item If $j \notin \tilde{S}_{[0]}^* \cup S_{[0]}^*$, similar to the proof of Proposition \ref{prop: sure screening joint contribution}, to guarantee Assumption \ref{asmp: sure screening iterative}, we need
		\begin{align}
			B_2\bp(S_{11}^{[1]} \supseteq S_{[0]}^* \cup \tilde{S}_{[0]}^* \cup \{j\}) &\geq B_2\bp(S_{11}^{[1]} \cap S^* =  S_{[0]}^* \cup \tilde{S}_{[0]}^* \cup \{j\}) \gtrsim 1, \label{eq: b2 iter lower bdd}\\
			B_2\bp(|\smj{} \cap S^*| > |S_{[0]}^* \cup  \tilde{S}_{[0]}^*|) &\lesssim B_2\bp(|S_{11}^{[1]} \cap S^*| > |S_{[0]}^* \cup  \tilde{S}_{[0]}^*|) \lesssim 1.\label{eq: b2 iter upper bdd}
		\end{align}
		Note that 
		\begin{align}
			\bp(|S_{11}^{[1]} \cap S^*| > |S_{[0]}^* \cup  \tilde{S}_{[0]}^*|) &\leq \sum_{i \in S_{[1]}^* \backslash \tilde{S}_{[0]}^*} \bp(S_{11}^{[1]} \cap S^* =  S_{[0]}^* \cup \tilde{S}_{[0]}^* \cup \{i\}) \\
			&\lesssim \bp(S_{11}^{[1]} \cap S^* =  S_{[0]}^* \cup \tilde{S}_{[0]}^* \cup \{j\}),
		\end{align}
		which combined with \eqref{eq: b2 iter lower bdd} and \eqref{eq: b2 iter upper bdd} leads to $B_2\asymp (\bp(S_{11}^{[1]} \cap S^* =  S_{[0]}^* \cup \tilde{S}_{[0]}^* \cup \{j\}))^{-1}$. Then because of \eqref{eq: eta lower bdd}, it holds that
		\begin{align}
			\bp(S_{11}^{[1]} \cap S^* =  S_{[0]}^* \cup \tilde{S}_{[0]}^* \cup \{j\}) &\geq D^{-1}\sum_{d= |S_{[0]}^* \cup \tilde{S}_{[0]}^*|+1}^D \prod_{i \in S_{[0]}^* \cup \tilde{S}_{[0]}^* \cup \{j\}} \tilde{\eta}_j^{[0]} \\
			&\gtrsim D^{-1}\sum_{d= |S_{[0]}^* \cup \tilde{S}_{[0]}^*|+1}^D D^{-|S^*_{[0]}|}(\log p)^{-|\tilde{S}_{[0]}^*|}p^{-1} \\
			&\gtrsim D^{-|S^*_{[0]}|}(\log p)^{-|\tilde{S}_{[0]}^*|}p^{-1},
		\end{align}
		leading to $B_2 \lesssim D^{|S^*_{[0]}|}(\log p)^{|\tilde{S}_{[0]}^*|}p \leq D^{|S^*_{[0]}|}(\log p)^{|S^*_{[1]}|-1}p$.
	\item If $j \in \tilde{S}_{[0]}^* \cup S_{[0]}^*$ and there exists $j_0 \in S^*_{[1]} \backslash (\tilde{S}_{[0]}^* \cup S_{[0]}^*)$, then
		\begin{align}
			B_2\bp(|\smj{} \cap S^*| \geq |S_{[0]}^* \cup  \tilde{S}_{[0]}^*|) &\lesssim B_2\bp(|S_{11}^{[1]} \cap S^*| > |S_{[0]}^* \cup  \tilde{S}_{[0]}^*|) \\
			&\lesssim B_2 \sum_{i \in S^*_{[1]} \backslash \tilde{S}_{[0]}^*}\bp(S_{11}^{[1]} \cap S^* = S_{[0]}^* \cup  \tilde{S}_{[0]}^* \backslash \{j\} \cup \{i\}) \\
			&\lesssim B_2\bp(S_{11}^{[1]} \cap S^* = S_{[0]}^* \cup  \tilde{S}_{[0]}^*\backslash \{j\} \cup \{j_0\}),
		\end{align}
		 and 
		 \begin{align}
		 	B_2\bp(S_{11}^{[1]} \supseteq S_{[0]}^* \cup \tilde{S}_{[0]}^*) &\geq B_2\bp(S_{11}^{[1]} \cap S^* =  S_{[0]}^* \cup \tilde{S}_{[0]}^*),
		 \end{align}
		 leading to $ (\bp(S_{11}^{[1]} \cap S^* =  S_{[0]}^* \cup \tilde{S}_{[0]}^*)^{-1} \lesssim B_2 \lesssim (\bp(S_{11}^{[1]} \cap S^* = S_{[0]}^* \cup  \tilde{S}_{[0]}^*\backslash \{j\} \cup \{j_0\}))^{-1}$. Similar to the previous calculation, it can be shown that $(\bp(S_{11}^{[1]} \cap S^* = S_{[0]}^* \cup  \tilde{S}_{[0]}^*\backslash \{j\} \cup \{j_0\}))^{-1} \lesssim D^{|S^*_{[0]}|}(\log p)^{|S^*_{[1]}|-1}p$.
	\item If $j \in \tilde{S}_{[0]}^* \cup S_{[0]}^*$ and $S^*_{[1]} \backslash (\tilde{S}_{[0]}^* \cup S_{[0]}^*) = \emptyset$, we have $\bp(|\smj{} \cap S^*| \geq |S_{[0]}^* \cup  \tilde{S}_{[0]}^*|) = 0$ because $S_{[0]}^* \cup  \tilde{S}_{[0]}^* = S^*$. Therefore it is sufficient to require $B_2 \gtrsim (\bp(S_{11}^{[1]} \cap S^* =  S_{[0]}^* \cup \tilde{S}_{[0]}^*))^{-1} = (\bp(S_{11}^{[1]} \supseteq S^*))^{-1}$, where $\bp(S_{11}^{[1]} \supseteq S^*) \gtrsim (D\wedge \log p)^{-p^*}$, implying that there exists $B_2 \lesssim (D \vee \log p)^{p^*} \vee p = p$ satisfying Assumption \ref{asmp: sure screening iterative}.
\end{itemize}
Above all, we complete the proof of Proposition \ref{prop: iter b2}.

\end{appendices}

\end{document}